\documentclass[a4paper,fleqn]{cas-dc}
\usepackage[numbers,sort&compress,longnamesfirst]{natbib}
\usepackage{setspace}
\usepackage{tcolorbox}
\tcbuselibrary{skins, breakable, theorems}
\usepackage{tabularx}  
\usepackage{diagbox}   
\usepackage{booktabs}  
\usepackage{multirow}  
\usepackage{pifont}   
\newcommand{\cmark}{\ding{51}}
\newcommand{\xmark}{\ding{55}}
\usepackage{hyphenat} 
\usepackage{array}     
\usepackage{lipsum}
\usepackage{microtype}
\usepackage{graphicx}
\usepackage{subcaption} 
\usepackage{amsmath,amssymb}
\usepackage{enumitem} 
\usepackage{hyperref} 
\usepackage{amsthm}    
\usepackage{hyperref}  
\allowdisplaybreaks[4] 
\usepackage[linesnumbered,ruled,vlined]{algorithm2e} 
\usepackage{multicol}     
\SetKwProg{Proc}{Procedure}{}{}
\SetKw{KwRet}{Return}
\SetKw{KwParse}{Parse}
\SetKwIF{If}{ElseIf}{Else}{If}{Then}{ElseIf}{Else}{EndIf}
\SetKwFor{For}{For}{Do}{EndFor}
\newtheoremstyle{mydefinition}  
  {0}                        
  {0}                        
  {\itshape}                    
  {2em}                         
  {\bfseries}                   
  {.}                           
  {0.5em}                       
  {\thmname{#1}\thmnumber{ #2} (\thmnote{\bfseries{#3}})}  
 
\theoremstyle{mydefinition}
\newtheorem{mydefinition}{Definition}[section]  
\newtheoremstyle{mytheorem}  
  {0}                        
  {0}                        
  {\itshape}                    
  {2em}                         
  {\bfseries}                   
  {.}                           
  {0.5em}                       
  {\thmname{#1}\thmnumber{ #2}}  
 
\theoremstyle{mytheorem}
\newtheorem{mytheorem}{Theorem}

\makeatletter
\renewenvironment{proof}[1][Proof]{%
  \par\pushQED{\qed}%
  \normalfont\topsep6\p@\@plus6\p@\relax
  \trivlist
  \item[\hskip\labelsep\hspace{2em}\itshape #1\@addpunct{.}]%
}{%
  \popQED\endtrivlist\@endpefalse
}
\makeatother
\def\tsc#1{\csdef{#1}{\textsc{\lowercase{#1}}\xspace}}
\tsc{WGM}
\tsc{QE}

\begin{document}
\let\WriteBookmarks\relax
\def\floatpagepagefraction{1}
\def\textpagefraction{.001}
\let\printorcid\relax 

\shorttitle{}    

\shortauthors{Yi Liang et al.}

\title[mode = title]{LTRAS: A Linkable Threshold Ring Adaptor Signature Scheme for Efficient and Private Cross-Chain Transactions}  

\author[1]{Yi Liang}
\author[1,2]{Jinguang Han}
\cormark[1]

\address[1]{School of Cyber Science and Engineering, Southeast University, Nanjing 211189, China} 
\address[2]{State Key Laboratory of Internet Architecture, Tsinghua University, Beijing, China, 100084} 
\ead{jghan@seu.edu.cn}
\cortext[1]{Corresponding author}  

\begin{abstract}
Despite the advantages of decentralization and immutability, blockchain technology faces significant scalability and throughput limitations, which has prompted the exploration of off-chain solutions like payment channels. Adaptor signatures have been considered a promising primitive for constructing such channels due to their support for atomicity, offering an alternative to traditional hash-timelock contracts. However, standard adaptor signatures may reveal signer identity, raising potential privacy concerns. While ring signatures can mitigate this issue by providing anonymity, they often introduce high communication overhead, particularly in multi-account payment settings commonly used in UTXO-based blockchains like Monero. To address these limitations, we propose a Linkable Threshold Ring Adaptor Signature (LTRAS) scheme, which integrates the conditional binding of adaptor signatures, the multi-account payment of threshold ring signatures, and the linkability for preventing double-spending. The formal definition, security model and concrete construction of LTRAS are provided. We also analyze its security and evaluate its performance through theoretical analysis and experimental implementation. Experimental results demonstrate that our scheme achieve significantly lower computation and communication overhead compared to existing schemes in large ring sizes and multi-account payment scenarios.
Finally, we discuss its application in cross-chain atomic swaps, demonstrating its potential for enhancing privacy and efficiency in blockchain transactions.
\end{abstract}



\begin{keywords}
Adaptor signature \sep 
Ring signature \sep 
Blockchain privacy \sep
Cross-chain atomic swaps \sep
Multi-account payment
\end{keywords}

\maketitle

\section{Introduction}
\label{sec:introduction}
Blockchain technology, as a decentralized distributed ledger system, has demonstrated significant potential in recent years across domains 
such as cryptocurrency \cite{yuan2018blockchain}, healthcare \cite{mcghin2019blockchain}, and supply chain \cite{kawaguchi2019application}. However, as blockchain technology evolves and its use cases expand, poor scalability and low throughput have emerged as critical bottlenecks. Off-chain payment channels are a key solution to this, enhancing overall transaction efficiency by alleviating the main blockchain's load through processing transactions off-chain. Nevertheless, the complex transaction scripting languages used in traditional 
payment channel schemes introduce new problems, not only increasing transaction fees for on-chain operations but also causing incompatibility 
with blockchains that have script limitations, thereby restricting broader adoption.

To address these challenges, adaptor signature  schemes \cite{poelstra2017scriptless,fournier2019one,aumayr2021generalized} have been proposed. This approach extends standard digital signatures by incorporating a hard mathematical relation (e.g., the discrete logarithm problem). A pre-signatureis generated with respect to this problem. Crucially, it can only be transformed into a valid full signatureby a party who knows a secret witnessto the relation, a property termed adaptability. Furthermore, the scheme guarantees witness extractability: anyone who sees both the pre-signature and the corresponding full signature can compute the witness. In a typical application, two parties exchange pre-signatures. Once the party holding the witness uses it to adapt one pre-signature, the other party can extract that witness and, in turn, adapt their own, ensuring a fair exchange. This conditional binding mechanism makes adaptor signatures suitable for scenarios 
where atomicity (or fairness) guarantees is required. Moreover, compared to traditional hash time-lock schemes, adaptor signatures rely solely on 
digital signature scripts that blockchain platforms inherently support, eliminating the need for complex hash or time-lock mechanisms, 
thereby significantly improving compatibility and efficiency. Currently, adaptor signatures are widely applied in numerous blockchain scenarios, 
including payment channel networks \cite{wang2024anonymity} and cross-chain atomic swaps \cite{fan2025attribute,yin2023survey}.

However, adaptor signatures do not address privacy issues: during signature generation and verification, 
the signer's account address is directly exposed on the blockchain. Attackers can analyze on-chain data to construct transaction graphs, 
inferring user identities, fund flows, and even business behavior patterns, leading to fraud or malicious marketing, thereby threatening 
asset security and leaking business secrets. Therefore, privacy-preserving adaptor signatures must be considered.

Ring signatures are privacy-preserving cryptographic schemes \cite{yuen2020ringct, yuen2021dualring,duan2024concise}. To address  privacy issues in Blockchain, Monero introduced ring signatures to hide the signer's public key within a "ring" containing other users' public keys, achieving anonymization of the signer's wallet address. 
 Linkable ring signatures further enhance security by allowing detection of whether the same signer used the ring in different 
 transactions without revealing identity, effectively preventing abuses such as double-spending. Schemes extending ring signatures to ring adaptor 
 signatures can resolve the privacy leakage issues of adaptor signatures, and ring adaptor signatures enables cross-chain atomic swap 
 protocols to be compatible with blockchains like Monero that utilize ring signatures. In transaction scenarios of cryptocurrencies like Monero, 
 the accounting model used is UTXO (Unspent Transaction Output), where users typically have multiple one-time public keys (i.e., account addresses), 
 and payments often require multiple accounts as inputs. If standard ring signatures are employed to sign transactions, each payment account must 
 be hidden in a separate ring, leading to a dramatic increase in signature size and higher communication and storage overhead. Compared to standard 
 ring signatures, threshold ring signatures \cite{duan2024concise}, many-out-of-many proofs \cite{diamond2021many} and 
 any-out-of-many proofs \cite{zheng2023leaking} can hide multiple accounts within the same ring, significantly reducing signature size, 
 thus offering advantages in multi-account joint payment scenarios. However, existing schemes have not addressed how to integrate threshold ring signatures 
 and linkability within the adaptor signature framework to meet the demands of multi-account joint payments in payment channel networks and 
 cross-chain atomic swaps, particularly in contexts requiring simultaneous privacy, efficiency, and atomicity. To achieve efficient and privacy-preserving blockchain applications, it is crucial to develop a linkable threshold ring adaptor signature scheme. This scheme aims to amalgamate the conditional binding of adaptor signatures, the multi-account efficiency inherent in threshold ring signatures, and the security features offered by linkable ring signatures.

\subsection{Our Contributions}
\label{sec:contributions}
In this paper, we propose a linkable threshold ring adaptor 
signature scheme (LTRAS). The characteristics of the scheme are as follows:
\begin{itemize}
	\item \textbf{Adaptability and Extractability}: The adaptability property ensures that a pre-signature can be converted into a full signature only by a user in possession of the correct witness. Conversely, the extractability property guarantees that the witness for the underlying hard relation can be derived by anyone who obtains both the pre-signature and the corresponding full signature.
	\item \textbf{Anonymity}: This property ensures that the signer's actual public key is hidden among a set of other public keys (the ring). A verifier of the signature cannot discern which specific public key in the ring corresponds to the genuine signer, thus guaranteeing the signer's anonymity.
	\item \textbf{Support for Joint Payments}: A signer can conceal multiple public keys within the same ring and subsequently complete a transfer transaction with a single signing operation. This is particularly applicable to joint payment scenarios involving multiple accounts in blockchain transactions.
	\item \textbf{Linkability}: The core functionality of linkability is to enable the detection of signatures generated by the same private key. This allows anyone to determine that different transactions were signed by the same entity, thereby effectively preventing double-spending attacks.
\end{itemize}

We also provide a formal definition, concrete construction, security models, and security proofs for the proposed scheme. 
Furthermore, we evaluate the performance of our scheme through theoretical analysis and experimental implementation, demonstrating its efficiency and practicality. 
Finally, we discuss its applications in cross-chain atomic swaps, showcasing its potential for real-world blockchain scenarios.

\subsection{Organization}
\label{sec:organization}
The paper proceeds in the following sequence. We first introduce the related work and preliminaries in Sections \ref{sec:relatedwork} and \ref{sec:preliminaries}, respectively. Section \ref{sec:LTRAS} subsequently presents our LTRAS scheme, including its formal definition, construction, security models, and proof. The performance of the scheme is then evaluated in Section \ref{sec:performance}. Furthermore, Section \ref{sec:application} discusses its potential applications in cross-chain atomic swaps. A conclusion is provided in Section \ref{sec:conclusion}.

\section{Related Work}
\label{sec:relatedwork}
The notion of adaptor signatures was initially introduced by Poelstra in 2017 \cite{poelstra2017scriptless}, though it lacked a formalized definition and security framework. Aiming to address this, Fournier \cite{fournier2019one} later conceptualized them in 2019 as one-time verifiably encrypted signatures (VES), building upon the VES concept by Boneh et al. \cite{boneh2003aggregate} and establishing a corresponding security model. A significant formalization advance was made by Aumayr et al. in 2021 \cite{aumayr2021generalized}, who elevated adaptor signatures to a standalone cryptographic primitive. They provided a rigorous definition and introduced the now-prevalent security model. Furthermore, we construct concrete schemes based on Schnorr and ECDSA signatures, two signature schemes widely used in blockchain systems, and provide corresponding security proofs, they laid a crucial theoretical groundwork for the broader application of adaptor signatures in blockchain technology.
Subsequently, threshold, \cite{wang2025tass, 
baecker2025fair}, identity-based \cite{zhu2023two,bao2023identity}, 
post-quantum \cite{tairi2021post, jana2024compact}, 
and multi-party \cite{kajita2024consecutive} adaptor signature schemes emerged. 
Further developments include the work of Liu et al. \cite{liu2024adaptor}, who noted that witness exposure limits applications and thus built witness-hiding adaptor signatures from digital signatures and a weak trapdoor commitment. Gerhart et al. \cite{gerhart2024foundations} highlighted vulnerabilities in the established security model by Aumayr et al. \cite{aumayr2021generalized}, proposing corrected definitions that ensure adaptability and witness extraction for all valid pre-signatures. Vanjani et al. \cite{vanjani2024functional} pointed out the lack of granularity in standard adaptor signatures and introduced functional adaptor signatures, wherein a function of the witness is extracted instead of the full witness. This scheme enables finer-grained access 
control to sensitive data while ensuring transaction fairness.

None of the aforementioned research work on adaptor signatures 
has addressed the privacy leakage issues inherent in adaptor 
signatures. In 2023, Qin et al. \cite{qin2023blindhub} constructed an ECDSA blind 
adaptor signature scheme, but this scheme only provides weak 
blindness. Although the signer cannot learn the specific content 
of the signed message, once the message-signature pair is 
published, the signer can link the signing process to the final 
signature. 
Moreover, blind adaptor signature schemes only protect the 
privacy of the signed message but not the privacy of the signer's 
identity. 

In 2020, Yuen et al. \cite{yuen2020ringct} proposed a linkable ring 
signature that combines inner product arguments to reduce signature 
size. Building upon this, Qin et al. \cite{qin2021generic} subsequently presented the first linkable ring adaptor signature scheme to protect signer's account address privacy and ensure anonymity. However, a key limitation of their scheme is that it can only hide a single account within a ring. Moreover, its signature size increases significantly with the number of payment accounts, making it impractical for scenarios involving multiple accounts. Separately, Sui et al. \cite{sui2022monet} devised a two-party sequential ring adaptor signature scheme, yet the lack of linkability prevents its application in typical cryptocurrency settings. 

Yuen et al. \cite{yuen2021dualring} introduced DualRing as a novel framework for constructing ring signatures, utilizing a sum argument protocol to achieve a logarithmic-scale signature size. Building upon this foundation, Wang et al. \cite{wang2024anonymity} presented a generic construction for linkable ring adaptor signatures. However, this scheme is also not suitable for multi-account scenarios. 

Duan et al. \cite{duan2024concise} hid multiple accounts within the same ring and 
utilized their proposed sliding window transformation protocol 
to construct a linkable threshold ring signature scheme that 
supports multi-account joint payment scenarios. 
Zheng et al. \cite{zheng2023leaking} proposed any-out-of-many 
proofs and designed a novel ring signature scheme capable 
of hiding multiple payment accounts within the same ring 
without revealing the exact number of accounts, thereby 
achieving stronger anonymity. However, both \cite{duan2024concise} 
and \cite{zheng2023leaking} only focus on ring signatures and do not address adaptor signatures. Liu et al. \cite{liu2021aucswap} presented AucSwap, a decentralized protocol for cross-chain asset transfer that employs an auction model. It employs a combination of atomic swap and an optimized Vickrey auction for efficient settlement. However, AucSwap relies on hash-timelock contracts for atomicity and cannot provide the privacy-preserving conditional binding capability critical for cross-chain atomic swaps. 
Scala et al. \cite{Scala2024zeromt} proposed ZeroMT, which optimizes intra-chain multi-transfer privacy using aggregated zero-knowledge proofs but is unable to support the conditional release of assets that is critical for cross-chain atomic swaps. Liu et al. \cite{LIU2025latticeObCPay} introduced a conditional payment scheme that is quantum-resistant and oracle-based, employing adaptor-inspired logic to ensure one-wayness and verifiability. This scheme remains confined to conditional payments initiated by oracles and offers no cross-chain support. 

\begin{table*}[h!]
\centering
\caption{Comparison of existing adaptor signature and ring signature schemes}
\label{tab:relatedwork_comparison}
\begin{tabular}{lccccc}
\toprule
\multirow{2}{*}{Schemes} & \multicolumn{5}{c}{Properties} \\
\cmidrule{2-6}
& Adaptability & Witness Extractability & Anonymity & Linkability & Joint Payments \\
\midrule
\cite{aumayr2021generalized}   & \cmark & \cmark & \xmark & \xmark & \xmark \\
\cite{vanjani2024functional}   & \cmark & \cmark & \xmark & \xmark & \xmark \\
\cite{qin2021generic}   & \cmark & \cmark & \cmark & \cmark & \xmark \\
\cite{sui2022monet}   & \cmark & \cmark & \cmark & \xmark & \xmark \\
\cite{wang2024anonymity}   & \cmark & \cmark & \cmark & \cmark & \xmark \\
\cite{duan2024concise}   & \xmark & \xmark & \cmark & \cmark & \cmark \\
\cite{zheng2023leaking}   & \xmark & \xmark & \cmark & \cmark & \cmark \\
Ours   & \cmark & \cmark  & \cmark  & \cmark  & \cmark \\
\bottomrule
\end{tabular}
\end{table*}

able \ref{tab:relatedwork_comparison} presents a comparative overview of the key characteristics between our approach and existing ones, providing a clear and systematic analysis.

\section{Preliminaries}
\label{sec:preliminaries}
\subsection{Notations}
For clarity, a summary of the notations employed in this work is compiled in Table \ref{tab:symbols} for reference.

\begin{table*}[h!]
\centering
\caption{Description of Symbols}
\label{tab:symbols}
\begin{tabular}{>{\raggedright\arraybackslash}p{3cm}>{\raggedright\arraybackslash}p{9cm}}
\toprule
\textbf{Symbol} & \textbf{Description} \\
\midrule
$\lambda$ & Security parameter \\
$\mathbb{G}_{p}$ & Commutative group with the order $p$ \\
$\mathbb{Z}_{p}$ & The finite field of integers modulo $p$ \\
$sk, pk$ & Secret key and public key \\
$\boldsymbol{PK}$ & A set with $n$ public keys $\{pk_0, \dots, pk_{n-1}\}$ \\
$|\boldsymbol{PK}|$ & The number of elements in the set $\boldsymbol{PK}$, i.e., $n$ \\
$H$ & A collision-resistant hash function which outputs an element in $\mathbb{Z}_{p}$ \\
$H(\boldsymbol{PK})$ & Hash function applied to $\boldsymbol{PK}$, i.e., $H(pk_0, \dots, pk_{n-1})$ \\
$[0,n-1]$ & A set of integers $\{0,1,\dots,n-1\}$ \\
$\mathcal{G}$ & A probabilistic polynomial-time algorithm that generates public parameters \\
$\epsilon$ & A negligible function \\
\bottomrule
\end{tabular}
\end{table*}

\subsection{Discrete Logarithm Assumption}
\begin{mydefinition}[Discrete Logarithm Assumption]
Under the discrete logarithm (DL) assumption, for a group $\mathbb{G}_{p}$ of prime order $p$ with generator $g$, it is computationally intractable for any probabilistic polynomial-time (PPT) adversary $\mathcal{A}$ to retrieve the exponent $x \in \mathbb{Z}_{p}$ from a given group element $y = g^x \in \mathbb{G}_{p}$. Formally, we define the advantage of $\mathcal{A}$ as:
\[
\text{Adv}_{\mathcal{A}}^{DL}(\lambda) = \Pr\left[\mathcal{A}(g, y) = x \mid y = g^x, x \in \mathbb{Z}_{p}\right].
\]
The assumption is considered to hold if $\text{Adv}_{\mathcal{A}}^{DL}(\lambda)\le \epsilon(\lambda)$ in the security parameter $\lambda$ for all PPT adversaries $\mathcal{A}$.
\end{mydefinition}

\subsection{Hard Relation}
A relation $\mathcal{R}$ is defined over statement-witness pairs $(W, w)$. Its language is $\mathcal{L}_{\mathcal{R}} = \{ W \mid \exists w : (W, w) \in \mathcal{R} \}$.
We call $\mathcal{R}$ a hard relation \cite{erwig2021two} if it satisfies three properties:
\begin{itemize}
	\item \textbf{Efficient Sampling:} There exists a probabilistic polynomial-time algorithm $\mathbf{GenR}(1^\lambda)$ that outputs a pair $(W,w) \in \mathcal{R}$.
	\item \textbf{Efficient Verification:} Membership in $\mathcal{R}$ is decidable in polynomial time.
	\item \textbf{Computational Hardness:} For any PPT algorithm $\mathcal{A}$, given $W \in \mathcal{L}_{\mathcal{R}}$ (where $(W,w) \leftarrow \mathbf{GenR}(1^\lambda)$), the probability that $\mathcal{A}$ outputs a valid witness $w'$ such that $(W, w') \in \mathcal{R}$ is negligible in $\lambda$.
\end{itemize}

\subsection{Digital Signature}
\label{sec:digital_signature}
A digital signature scheme $\Sigma$ comprises the following probabilistic algorithms:
\begin{itemize}
	\item $\mathbf{Setup}(1^\lambda) \rightarrow pp$: This algorithm takes the security parameter $\lambda$ and returns the system's public parameters $pp$.
	\item $\mathbf{KeyGen}(pp) \rightarrow (pk,sk)$: On input $pp$, it generates a pair consisting of a public verification key $pk$ and a private signing key $sk$.
	\item $\mathbf{Sign}(pp,sk,m) \rightarrow \sigma$: For a given message $m$, it uses the secret key $sk$ and parameters $pp$ to produce a signature $\sigma$.
	\item $\mathbf{Verify}(pp,pk,m,\sigma) \rightarrow 0/1$: This deterministic algorithm checks the validity of signature $\sigma$ for message $m$ under public key $pk$, outputting 1 if valid and 0 otherwise.
\end{itemize}
A secure signature scheme $\Sigma$ guarantees:
\begin{itemize}
	\item \textbf{Completeness}: For all key pairs $(pk, sk)$ generated by $\mathbf{KeyGen}(pp)$, every signature $\sigma \leftarrow \mathbf{Sign}(pp,sk,m)$ satisfies $\mathbf{Verify}(pp,pk,m,\sigma) = 1$ with overwhelming probability.
	\item \textbf{Unforgeability}: It must be infeasible for any polynomial-time adversary $\mathcal{A}$ who does not possess the secret key $sk$ to produce a valid signature on any (potentially new) message.
\end{itemize}

\subsection{Adaptor Signature}
An adaptor signature scheme associated with a hard relation $\mathcal{R}$ and a standard digital signature scheme $\Sigma= \left(\mathbf{Setup}, \mathbf{KeyGen}, \mathbf{Sign}, \mathbf{Verify}\right)$ comprises a suite of four algorithms, denoted as $\Pi_{\mathcal{R},\Sigma}=\left(\mathbf{PreSign},\mathbf{PreVerify}, \right.$ $\left.\\ \mathbf{Adapt},\mathbf{Ext}\right)$. As detailed in Section \ref{sec:digital_signature}, $\Sigma$ consists of four algorithms, briefly recalled below: $\mathbf{Setup}(1^\lambda) \rightarrow pp$, $\mathbf{KeyGen}(pp) \rightarrow (pk,sk)$, $\mathbf{Sign}(pp,sk,m) \rightarrow \sigma$, and $\mathbf{Verify}(pp,pk,m,\sigma) \rightarrow 0/1$.
The algorithms defining $\Pi_{\mathcal{R},\Sigma}$ are as follows:
\begin{itemize}
	\item $\mathbf{PreSign}(pp,pk,sk,m,W) \rightarrow \tilde{\sigma}$: This algorithm accepts the public parameters $pp$, a public key $pk$, the corresponding secret key $sk$, a message $m$, and a statement $W$ belonging to the hard relation. It yields a pre-signature $\tilde{\sigma}$.
	\item $\mathbf{PreVerify}(pp,pk,\tilde{\sigma},m,W) \rightarrow 0/1$: On input $pp$, $pk$, a pre-signature $\tilde{\sigma}$, a message $m$, and a statement $W$, this deterministic algorithm returns 1 if $\tilde{\sigma}$ is valid, and 0 otherwise.
	\item $\mathbf{Adapt}(pp,\tilde{\sigma},w) \rightarrow \sigma$: Taking the public parameters $pp$, a pre-signature $\tilde{\sigma}$, and a witness $w$, this algorithm produces a full signature $\sigma$.
	\item $\mathbf{Ext}(pp,W,\tilde{\sigma},\sigma) \rightarrow w/\bot$: Given $pp$, a statement $W$, a pre-signature $\tilde{\sigma}$, and a full signature $\sigma$, this algorithm outputs the corresponding witness $w$ if $(W,w) \in \mathcal{R}$; otherwise, it returns $\bot$.
\end{itemize}

An adaptor signature scheme $\Pi_{\mathcal{R},\Sigma}$ is defined to satisfy the following four properties:
\begin{itemize}
	\item \textbf{Correctness}: For any security parameter $\lambda \in \mathbb{N}$, any message $m \in \{0,1\}^*$, and any valid statement/witness pair $(W,w) \in \mathcal{R}$, the following condition holds:
	\[
	\text{Pr}\left[ \begin{array}{c}
		\mathbf{PreVerify}(pp,\\ pk,\tilde{\sigma}, m,W)=1
		\\ \land \mathbf{Verify} (pp, pk,\\ \sigma,m)=1
		\\ \land (W,w') \in \mathcal{R}
	\end{array}
	\left| \begin{array}{c}
		pp \leftarrow \mathbf{Setup}(1^\lambda),\\
		(pk,sk) \leftarrow \mathbf{KeyGen}(pp),\\
		\tilde{\sigma} \leftarrow \mathbf{PreSign}(pp,pk,\\ sk,m, W),
		\sigma \leftarrow \mathbf{Adapt}\\(pp, \tilde{\sigma},w),
		w' \leftarrow \mathbf{Ext}\\ (pp, W,\tilde{\sigma},\sigma)
	\end{array} \right. \right] =1.
	\]
	\item \textbf{Pre-signature Adaptability}: For all $\lambda \in \mathbb{N}$, all messages $m \in \{0,1\}^*$, all bitstrings $\tilde{\sigma}$, and all pairs $(W,w) \in \mathcal{R}$, if $\tilde{\sigma}$ is a valid pre-signature, then adapting it with witness $w$ yields a standard signature that verifies. Formally:
	\[
	\text{Pr}\left[ \begin{array}{c}
    \mathbf{Verify}(pp, pk,\\ \sigma,m)=1
	\end{array}
	\left| \begin{array}{c}
		pp \leftarrow \mathbf{Setup}(1^\lambda),\\
		(pk,sk) \leftarrow  \mathbf{KeyGen}(pp),\\
		\mathbf{PreVerify}(pp,pk,\\ \tilde{\sigma},m,W)=1,\\
		\sigma \leftarrow \mathbf{Adapt}(pp,\tilde{\sigma},w)
	\end{array} \right. \right] =1.
	\]
	This guarantees an efficient conversion from a valid pre-signature to a full signature given the correct witness.

	\item \textbf{Witness Extractability}: This property requires that a witness can be extracted from any valid pre-signature and signature pair. No PPT adversary $\mathcal{A}$ can generate a convincing pair $(\tilde{\sigma}, \sigma)$ from which the extract algorithm $\mathbf{Ext}$ fails to recover a valid witness, except with probability negligible in $\lambda$.

	\item \textbf{Unforgeability}: The scheme resists signature forgery. Specifically, the success probability of any PPT adversary $\mathcal{A}$ in forging a valid signature $\sigma^*$ on a message $m^*$ without knowledge of the secret key $sk$ must be a negligible function in the security parameter $\lambda$.
\end{itemize}

\begin{algorithm}[h!]
	\caption{Experiment $\textbf{aWitExt}_{\mathcal{A},\Pi_{\mathcal{R},\Sigma}}$}
	\label{alg:WitExt}
	\Proc{$\mathbf{aWitExt}_{\mathcal{A},\Pi _{\mathcal{R},\Sigma}}(\lambda):$}{
	\parbox[t]{0.8\linewidth}{This security model is captured via a sequence of four games, wherein the challenger $\mathcal{C}$ interacts with the adversary $\mathcal{A}$.}\\
	\parbox[t]{0.8\linewidth}{$\mathsf{Setup}$. $\mathcal{C}$ runs $\mathbf{Setup}(1^\lambda) \to pp$, 
	initializes the signature query set $\boldsymbol{Q}=\emptyset$, 
	and returns $pp$ to $\mathcal{A}$. $\mathcal{A}$ runs $\mathbf{GenR}(pp) \to (\boldsymbol{W},w)$, 
	and returns $\boldsymbol{W}$ to $\mathcal{C}$.}\\
	\parbox[t]{0.8\linewidth}{$\mathsf{KeyGen}$. $\mathcal{C}$ runs $\mathbf{KeyGen}(pp) \to (pk_i,sk_i)$, where
	$i \in [0,n-1]$, sets $\boldsymbol{PK}=(pk_0,\dots,pk_{n-1})$, and returns 
	$\boldsymbol{PK}$ to $\mathcal{A}$.}\\
	\parbox[t]{0.8\linewidth}{$\mathsf{Query}$. The capabilities of the adversary $\mathcal{A}$ is modeled by 
	making queries to the following oracles:}\\
	\Proc{$\mathcal{SO}(\boldsymbol{PK},\boldsymbol{S},m):$}{
		\parbox[t]{0.75\linewidth}{On input a set of public keys $\boldsymbol{PK}$, 
		a set of indices $\boldsymbol{S}\subseteq [0,n-1]$, 
		and a message $m$, the oracle runs 
		$\mathbf{Sign}(pp,\boldsymbol{PK},\left\{ sk_i \right\} _{i\in \boldsymbol{S}},m) \rightarrow \sigma$, 
		sets $\boldsymbol{Q}=\boldsymbol{Q} \cup \left\{m \right\}$, 
		and returns the signature $\sigma$.}
	}
	\Proc{$\mathcal{PSO}(\boldsymbol{PK},\boldsymbol{S},m,\boldsymbol{W}):$}{
		\parbox[t]{0.75\linewidth}{On input a set of public keys $\boldsymbol{PK}$, 
		a set of indices $\boldsymbol{S}\subseteq [0,n-1]$, 
		a message $m$, and the the statement $\boldsymbol{W}$ from the hard relation, 
		the oracle runs $\mathbf{PreSign}(pp,$ $\boldsymbol{PK},\left\{ sk_i \right\} _{i\in \boldsymbol{S}},m,\boldsymbol{W}) \rightarrow \tilde{\sigma}$, 
		sets $\boldsymbol{Q}=$ $\boldsymbol{Q} \cup \left\{m \right\}$, 
		and returns the pre-signature $\tilde{\sigma}$.}
	}
	\parbox[t]{0.8\linewidth}{$\mathsf{Witness\ Extraction}$. $\mathcal{A}$ sends the target message $(\boldsymbol{S}^*,m^*,\boldsymbol{W}^*=(W_1^*,W_2^*))$ 
	to $\mathcal{C}$. The challenger $\mathcal{C}$ runs 
	$\mathbf{PreSign}(pp,\boldsymbol{PK},\left\{ sk_i \right\} _{i\in \boldsymbol{S}^*},$ $m^*,\boldsymbol{W}^*)\to \tilde{\sigma}^*$, 
	then sends $\tilde{\sigma}^*$ to $\mathcal{A}$. 
	$\mathcal{A}$ generates a signature $\sigma^*$ of $m^*$, and sends $\sigma^*$ to $\mathcal{C}$. 
	$\mathcal{C}$ runs $\mathbf{Ext}(pp,\boldsymbol{W^*},\tilde{\sigma}^*,\sigma^*)\to w^*$. The experiment outputs 1, indicating $\mathcal{A}$'s win, provided that:}
	\begin{itemize}
		\item $m^* \notin \boldsymbol{Q}$,
		\item $\mathbf{Verify}(pp,\boldsymbol{PK},\sigma^*,t,m^*)=1$,
		\item $(W_1^*,w^*)\notin \mathcal{R} \lor (W_2^*,w^*)\notin \mathcal{R}$.
	\end{itemize}
	\parbox[t]{0.8\linewidth}{In all other cases, the output is 0.}
	}
\end{algorithm}

\begin{algorithm}[h!]
	\caption{Experiment $\textbf{aSignForge}_{\mathcal{A},\Pi_{\mathcal{R},\Sigma}}$}
	\label{alg:SignForge}
	\Proc{$\mathbf{aSignForge}_{\mathcal{A},\Pi _{\mathcal{R},\Sigma}}(\lambda):$}{
	\parbox[t]{0.8\linewidth}{This security model is captured via a sequence of four games, wherein the challenger $\mathcal{C}$ interacts with the adversary $\mathcal{A}$.}\\
	\parbox[t]{0.8\linewidth}{$\mathsf{Setup}$. $\mathcal{C}$ runs $\mathbf{Setup}(1^\lambda) \to pp$, 
	$\mathbf{GenR}(pp) \to (\boldsymbol{W},w)$, initializes the corrupted party set 
	$\boldsymbol{F}=\emptyset$ and the signature query set $\boldsymbol{Q}=\emptyset$, 
	and returns $pp$ and $\boldsymbol{W}$ to $\mathcal{A}$.}\\
	\parbox[t]{0.8\linewidth}{$\mathsf{KeyGen}$. $\mathcal{C}$ runs $\mathbf{KeyGen}(pp) \to (pk_i,sk_i)$, where
	$i \in [1,q_k]$, sets $\boldsymbol{P}=(pk_1,pk_2,\dots,pk_{q_k})$, and returns 
	$\boldsymbol{P}$ to $\mathcal{A}$.}\\
	\parbox[t]{0.8\linewidth}{$\mathsf{Query}$. The capabilities of the adversary $\mathcal{A}$ is modeled by 
	making queries to the following oracles:}\\
	\Proc{$\mathcal{CO}(i):$}{
		\parbox[t]{0.75\linewidth}{On input an index $i \in [1,q_k]$, the oracle sets 
		$\boldsymbol{F}=\boldsymbol{F} \cup \left\{ pk_i \right\}$, 
		and returns the corresponding secret key $sk_i$.}
	}
	\Proc{$\mathcal{SO}(\boldsymbol{PK},\boldsymbol{S},m):$}{
		\parbox[t]{0.75\linewidth}{On input a set of public keys $\boldsymbol{PK}\subseteq \boldsymbol{P}$, 
		a set of indices $\boldsymbol{S}\subseteq [0,n-1]$, 
		and a message $m$, where $|\boldsymbol{PK}|=n$, the oracle runs 
		$\mathbf{Sign}(pp,\boldsymbol{PK},\left\{ sk_i \right\} _{i\in \boldsymbol{S}},m) 
		\rightarrow \sigma$, sets $\boldsymbol{Q}=\boldsymbol{Q} \cup \left\{(\boldsymbol{PK},\boldsymbol{S},m )\right\}$, 
		and returns the signature $\sigma$.}
	}
	\Proc{$\mathcal{PSO}(\boldsymbol{PK},\boldsymbol{S},m,\boldsymbol{W}):$}{
		\parbox[t]{0.75\linewidth}{On input a set of public keys $\boldsymbol{PK}\subseteq \boldsymbol{P}$, 
		a set of indices $\boldsymbol{S}\subseteq [0,n-1]$, 
		a message $m$, where $|\boldsymbol{PK}|=n$, and the the statement $\boldsymbol{W}$ from the hard relation, 
		the oracle runs $\mathbf{PreSign}(pp,\boldsymbol{PK},
		\left\{ sk_i \right\} _{i\in \boldsymbol{S}},m,\boldsymbol{W}) 
		\rightarrow \tilde{\sigma}$, sets $\boldsymbol{Q}=\boldsymbol{Q} \cup 
		\left\{(\boldsymbol{PK},\boldsymbol{S},m )\right\}$, 
		and returns the pre-signature $\tilde{\sigma}$.}
	}
	\parbox[t]{0.8\linewidth}{$\mathsf{Forgery}$. The adversary $\mathcal{A}$ sends the target message $(\boldsymbol{PK}^*,\boldsymbol{S}^*,m^*)$ to
	$\mathcal{C}$. The challenger $\mathcal{C}$ runs 
	$\mathbf{PreSign}(pp,\boldsymbol{PK}^*,\left\{ sk_i \right\} _{i\in \boldsymbol{S}^*},m^*,\boldsymbol{W})\to \tilde{\sigma}^*$, 
	then sends $\tilde{\sigma}^*$ to $\mathcal{A}$. 
	If $\mathcal{A}$ successfully outputs $(\boldsymbol{PK}^*,\sigma^*,m^*)$ 
	with the following conditions:}
	\begin{itemize}
		\item $\boldsymbol{PK}^*\subseteq \boldsymbol{P}$ and 
		$\left| \boldsymbol{PK}^* \right|=n$,
		\item $|\boldsymbol{PK}^*\cap \boldsymbol{F}|\le t-1$,
		\item $(\boldsymbol{PK}^*,*,m^*) \notin \boldsymbol{Q}$,
		\item $\mathbf{Verify}(pp,\boldsymbol{PK}^*,\sigma^*,t,m^*)=1$.
	\end{itemize}
	\parbox[t]{0.8\linewidth}{When these hold, the experiment returns 1; otherwise, it returns 0.}
	}
\end{algorithm}

\subsection{Sliding Window Transformation}
The sliding window transformation (SWT) \cite{duan2024concise} is designed to transform a standard ring signature scheme into a $(t,n)$-threshold variant within a homomorphic cryptosystem.

Let $\boldsymbol{PK}=\left\{pk_0,pk_1,\dots, pk_{n-1} \right\}$ be a set of $n$ public keys, and let $t$ be the specified threshold. The SWT procedure is performed as follows:
\begin{itemize}
	\item First, compute $d=H(\boldsymbol{PK})$, where $H$ is a hash function mapping the key set to a fixed-size digest. The inclusion of $d$ serves to mitigate rogue-key attacks.
	\item Then, for each index $i\in [0,n-1]$, calculate the aggregate of $t$ consecutive public keys beginning with $pk_i$, wrapping around the set as needed. Formally, this is:
	\[
		y_i = \prod_{k=i (mod\ n)}^{(i+t-1) (mod\ n)} pk_{k}^d.
	\]
	\item The output is the transformed set of public keys, denoted as $\boldsymbol{PK'}=\left\{y_0,y_1,\dots,y_{n-1}\right\}$.
\end{itemize}

\section{The Linkable Threshold Ring Adaptor Signature Scheme}
\label{sec:LTRAS}
\subsection{Formal Definition}
\begin{mydefinition}[Linkable $(t,n)$-Threshold Ring Adaptor Signature Scheme]
	A linkable $(t,n)$-threshold ring adaptor signature (LTRAS) scheme $\Pi _{\mathcal{R},\Sigma}=\left(
	\mathbf{Setup},\right.$ $\left.\\\mathbf{KeyGen}, \mathbf{GenR},\mathbf{PreSign},\mathbf{PreVerify},
	 \mathbf{Adapt},\mathbf{Verify},
	 \mathbf{Ext},\right.$ $\left.\\ \mathbf{Link}\right)$ is comprised of a tuple of algorithms operating relative to a hard relation $\mathcal{R}$ and a standard signature scheme $\Sigma$. The individual algorithms are specified as follows.
\end{mydefinition}
\begin{itemize}
	\item $\mathbf{Setup}(1^\lambda) \rightarrow pp$: This algorithm accepts the security parameter $\lambda$ and returns the system's public parameters $pp$.
	\item $\mathbf{KeyGen}(pp) \rightarrow (pk,sk)$: On input of the public parameters $pp$, the key generation algorithm produces a public/secret key pair $(pk, sk)$.
	\item $\mathbf{GenR}(pp) \rightarrow (\boldsymbol{W},w)$: Taking $pp$ as input, this algorithm samples two statement/witness pairs $(W_1,w)$ and $(W_2,w)$ from the hard relation $\mathcal{R}$. It outputs the combined statement $\boldsymbol{W}=(W_1,W_2)$ along with the common witness $w$.
	\item $\mathbf{PreSign}(pp,\boldsymbol{PK},\boldsymbol{SK},m,\boldsymbol{W}) 
	\rightarrow \tilde{\sigma}$: This algorithm accepts the public parameters $pp$, a set of public keys $\boldsymbol{PK}$, a set of secret keys
	$\boldsymbol{SK}$, message $m$, and the statement $\boldsymbol{W}$ from the hard
	relation, where $|\boldsymbol{PK}|=n$, and $|\boldsymbol{SK}|\ge t$, and outputs a pre-signature $\tilde{\sigma}$.
	\item $\mathbf{PreVerify}(pp,\boldsymbol{PK},\tilde{\sigma},t,m,\boldsymbol{W}) 
	\rightarrow 0/1$: On input $pp$, $\boldsymbol{PK}$, a pre-signature $\tilde{\sigma}$, 
	threshold $t$, a message $m$, and the statement $\boldsymbol{W}$ from the hard relation, 
	this deterministic algorithm returns 1 if $\tilde{\sigma}$ is valid, and 0 otherwise.
	\item $\mathbf{Adapt}(pp,\tilde{\sigma},w) \rightarrow \sigma$: Taking the public parameters $pp$, a pre-signature $\tilde{\sigma}$, and a witness $w$,
	this algorithm produces a full signature $\sigma$.
	\item $\mathbf{Verify}(pp,\boldsymbol{PK},\sigma,t,m) \rightarrow 0/1$: On input $pp$, $\boldsymbol{PK}$, a signature $\sigma$, the threshold $t$, and the message $m$, this deterministic algorithm returns 1 if $\sigma$ is valid, and 0 otherwise.
	\item $\mathbf{Ext}(pp,\boldsymbol{W},\tilde{\sigma},\sigma) \rightarrow w/\bot$: Given $pp$, the statement $\boldsymbol{W}=(W_1,W_2)$, the pre-signature $\tilde{\sigma}$, and the full signature $\sigma$, this algorithm outputs the corresponding witness $w$ if $(W_1,w) \in \mathcal{R}$ and $(W_2,w) \in \mathcal{R}$; otherwise, it returns $\bot$.
	\item $\mathbf{Link}(\sigma ', \sigma '') \rightarrow 0/1$: This algorithm determines whether two given signatures $\sigma '$ and $\sigma ''$ originate from the same signer. It returns 1 upon confirmation of linkage, and 0 otherwise.
\end{itemize}

\begin{mydefinition}[Correctness]
\label{def:correctness}
	We say that a linkable $(t,n)$-threshold ring adaptor signature scheme $\Pi_{\mathcal{R},\Sigma}$ exhibits correctness whenever the following is true for all security parameters $\lambda \in \mathbb{N}$, all messages $m \in \{0,1\}^*$, and all thresholds $t \in [1,n]$:
\end{mydefinition}
\[
	\text{Pr}\left[ \begin{array}{c}
		\mathbf{PreVerify}(pp,\\ \boldsymbol{PK}, \tilde{\sigma},t,m,\boldsymbol{W})\\=1
		\land \mathbf{Verify}(pp,\\ \boldsymbol{PK}, \sigma,t,m)=1\\
		\land (W_1,w') \in \mathcal{R} \\ \land (W_2,w')\in \mathcal{R}
	\end{array}
	\left| \begin{array}{c}
		pp \leftarrow \mathbf{Setup}(1^\lambda),\\
		(pk,sk) \leftarrow \mathbf{KeyGen}(pp),\\
		(\boldsymbol{W}=(W_1,W_2), w) \leftarrow\\ \mathbf{GenR}(pp),
		\tilde{\sigma}\leftarrow \mathbf{PreSign}\\(pp, \boldsymbol{PK},\boldsymbol{SK},m,\boldsymbol{W}),\\
		\sigma \leftarrow \mathbf{Adapt}(pp,\tilde{\sigma},w),\\
		w' \leftarrow \mathbf{Ext}(pp,\boldsymbol{W},\tilde{\sigma},\sigma)
	\end{array} \right. \right] =1.
\]

\begin{mydefinition}[Pre-signature Adaptability]
	We say that a linkable $(t,n)$-threshold ring adaptor signature scheme $\Pi_{\mathcal{R},\Sigma}$ exhibits pre-signature adaptability whenever the following is true for all security parameters $\lambda \in \mathbb{N}$, all messages $m \in \{0,1\}^*$, all thresholds $t \in [1,n]$, and all bitstrings $\tilde{\sigma}\in \{0,1\}^*$:
\end{mydefinition}
\[
	\text{Pr}\left[ \begin{array}{c}
    \mathbf{Verify}(pp,\boldsymbol{PK},\\ \sigma,t,m)=1
    \end{array}
    \left| \begin{array}{c}
		pp \leftarrow \mathbf{Setup}(1^\lambda),\\
		(pk,sk) \leftarrow \mathbf{KeyGen}(pp),\\
		(\boldsymbol{W},w) \leftarrow \mathbf{GenR}(pp),\\
		\mathbf{PreVerify}(pp,\boldsymbol{PK},\\ \tilde{\sigma},t,m,\boldsymbol{W})=1,\\
		\sigma \leftarrow \mathbf{Adapt}(pp,\tilde{\sigma},w)
	\end{array} \right. \right] =1.
\]

\begin{algorithm}[h!]
	\caption{Experiment $\textbf{aAnon}_{\mathcal{A},\Pi_{\mathcal{R},\Sigma}}$}
	\label{alg:Anon}
	\Proc{$\mathbf{aAnon}_{\mathcal{A},\Pi _{\mathcal{R},\Sigma}}(\lambda):$}{
	\parbox[t]{0.8\linewidth}{This security model is captured via a sequence of five games, wherein the challenger $\mathcal{C}$ interacts with the adversary $\mathcal{A}$.}\\
	\parbox[t]{0.8\linewidth}{$\mathsf{Setup}$. $\mathcal{C}$ runs $\mathbf{Setup}(1^\lambda) \to pp$, 
	$\mathbf{GenR}(pp) \to (\boldsymbol{W},w)$, initializes the signature query set $\boldsymbol{Q}=\emptyset$, 
	and returns $pp$ and $\boldsymbol{W}$ to $\mathcal{A}$.}\\
	\parbox[t]{0.8\linewidth}{$\mathsf{KeyGen}$. $\mathcal{C}$ runs $\mathbf{KeyGen}(pp) \to (pk_i,sk_i)$, where
	$i \in [1,q_k]$, sets $\boldsymbol{P}=(pk_1,pk_2,\dots,pk_{q_k})$, and returns 
	$\boldsymbol{P}$ to $\mathcal{A}$.}\\
	\parbox[t]{0.8\linewidth}{$\mathsf{Query}$. The capabilities of the adversary $\mathcal{A}$ is modeled by 
	making queries to the following oracles:}\\
	\Proc{$\mathcal{SO}(\boldsymbol{PK},\boldsymbol{S},m):$}{
		\parbox[t]{0.75\linewidth}{On input a set of public keys $\boldsymbol{PK}\subseteq \boldsymbol{P}$, 
		a set of indices $\boldsymbol{S}\subseteq [0,n-1]$, 
		and a message $m$, where $|\boldsymbol{PK}|=n$, the oracle runs 
		$\mathbf{Sign}(pp,\boldsymbol{PK},\left\{ sk_i \right\} _{i\in \boldsymbol{S}},m) 
		\rightarrow \sigma$, sets $\boldsymbol{Q}=\boldsymbol{Q} \cup 
		\left\{(\boldsymbol{PK},\boldsymbol{S},m )\right\}$, 
		and returns the signature $\sigma$.}
	}
	\Proc{$\mathcal{PSO}(\boldsymbol{PK},\boldsymbol{S},m,\boldsymbol{W}):$}{
		\parbox[t]{0.75\linewidth}{On input a set of public keys $\boldsymbol{PK}\subseteq \boldsymbol{P}$, 
		a set of indices $\boldsymbol{S}\subseteq [0,n-1]$, 
		a message $m$, where $|\boldsymbol{PK}|=n$, and the the statement $\boldsymbol{W}$ from the hard relation, 
		the oracle runs $\mathbf{PreSign}(pp,\boldsymbol{PK},
		\left\{ sk_i \right\} _{i\in \boldsymbol{S}},m,\boldsymbol{W}) 
		\rightarrow \tilde{\sigma}$, sets $\boldsymbol{Q}=\boldsymbol{Q} \cup \left\{(\boldsymbol{PK},\boldsymbol{S},m )\right\}$, 
		and returns the pre-signature $\tilde{\sigma}$.}
	}
	\parbox[t]{0.8\linewidth}{$\mathsf{Challenge}$. $\mathcal{A}$ sends the target message 
	$(\boldsymbol{PK}^*,\boldsymbol{S}_0^*,\boldsymbol{S}_1^*,m^*,\boldsymbol{W}^*=(W_1^*,W_2^*))$ 
	to $\mathcal{C}$, where $\boldsymbol{PK}^*\subseteq \boldsymbol{P}$, $|\boldsymbol{PK}^*|=n$, $|\boldsymbol{S}_0^*|=|\boldsymbol{S}_1^*|=t$, 
	$\boldsymbol{S}_0^*\ne \boldsymbol{S}_1^*$, $(*,\boldsymbol{S}_0^*,*) \notin \boldsymbol{Q}$, 
	and $(*,\boldsymbol{S}_1^*,*) \notin \boldsymbol{Q}$. The challenger $\mathcal{C}$ runs 
	$\mathbf{PreSign}(pp,\boldsymbol{PK}^*,$ $\left\{ sk_i \right\} _{i\in \boldsymbol{S}_b^*},m^*,\boldsymbol{W}^*)\to \tilde{\sigma}^*$, for a 
	randomly chosen bit $b\in \{0,1\}$, then sends $\tilde{\sigma}^*$ to $\mathcal{A}$.}\\
	\parbox[t]{0.8\linewidth}{$\mathsf{Output}$. $\mathcal{A}$ returns a bit $b' \in \{0,1\}$. The experiment outputs 1 if when $b'=b$ and outputs 0 otherwise.}\\
	}
\end{algorithm}

\begin{algorithm}[h!]
	\caption{Experiment $\textbf{aLink}_{\mathcal{A},\Pi_{\mathcal{R},\Sigma}}$}
	\label{alg:Link}
	\Proc{$\mathbf{aLink}_{\mathcal{A},\Pi _{\mathcal{R},\Sigma}}(\lambda):$}{
	\parbox[t]{0.8\linewidth}{This security model is captured via a sequence of four games, wherein the challenger $\mathcal{C}$ interacts with the adversary $\mathcal{A}$.}\\
	\parbox[t]{0.8\linewidth}{$\mathsf{Setup}$. $\mathcal{C}$ runs $\mathbf{Setup}(1^\lambda) \to pp$, 
	$\mathbf{GenR}(pp) \to (\boldsymbol{W},w)$, initializes the corrupted party set 
	$\boldsymbol{F}=\emptyset$ and the signature query set $\boldsymbol{Q}=\emptyset$, 
	and returns $pp$ and $\boldsymbol{W}$ to $\mathcal{A}$.}\\
	\parbox[t]{0.8\linewidth}{
	$\mathsf{KeyGen}$. $\mathcal{C}$ runs $\mathbf{KeyGen}(pp) \to (pk_i,sk_i)$, where
	$i \in [1,q_k]$, sets $\boldsymbol{P}=(pk_1,pk_2,\dots,pk_{q_k})$, and returns 
	$\boldsymbol{P}$ to $\mathcal{A}$.}\\
	\parbox[t]{0.8\linewidth}{
	$\mathsf{Query}$. The capabilities of the adversary $\mathcal{A}$ is modeled by 
	making queries to the following oracles:}\\
	\Proc{$\mathcal{CO}(i):$}{
		\parbox[t]{0.75\linewidth}{On input an index $i \in [1,q_k]$, the oracle sets 
		$\boldsymbol{F}=\boldsymbol{F} \cup \left\{ pk_i \right\}$, 
		and returns the corresponding secret key $sk_i$.}
	}
	\Proc{$\mathcal{SO}(\boldsymbol{PK},\boldsymbol{S},m):$}{
		\parbox[t]{0.75\linewidth}{On input a set of public keys $\boldsymbol{PK}\subseteq \boldsymbol{P}$, 
		a set of indices $\boldsymbol{S}\subseteq [0,n-1]$, 
		and a message $m$, where $|\boldsymbol{PK}|=n$, the oracle runs 
		$\mathbf{Sign}(pp,\boldsymbol{PK},\left\{ sk_i \right\} _{i\in \boldsymbol{S}},m) 
		\rightarrow \sigma$, sets $\boldsymbol{Q}=\boldsymbol{Q} \cup 
		\left\{(\boldsymbol{PK},\boldsymbol{S},m) \right\}$, 
		and returns the signature $\sigma$.}
	}
	\Proc{$\mathcal{PSO}(\boldsymbol{PK},\boldsymbol{S},m,\boldsymbol{W}):$}{
		\parbox[t]{0.75\linewidth}{On input a set of public keys $\boldsymbol{PK}\subseteq \boldsymbol{P}$, 
		a set of indices $\boldsymbol{S}\subseteq [0,n-1]$, 
		a message $m$, where $|\boldsymbol{PK}|=n$, and the the statement $\boldsymbol{W}$ from the hard relation, 
		the oracle runs $\mathbf{PreSign}(pp,\boldsymbol{PK},
		\left\{ sk_i \right\} _{i\in \boldsymbol{S}},m,\boldsymbol{W}) 
		\rightarrow \tilde{\sigma}$, sets $\boldsymbol{Q}=\boldsymbol{Q} \cup 
		\left\{(\boldsymbol{PK},\boldsymbol{S},m) \right\}$, 
		and returns the pre-signature $\tilde{\sigma}$.}
	}
	\parbox[t]{0.8\linewidth}{$\mathsf{Challenge}$. If $\mathcal{A}$ successfully outputs two signatures $(\boldsymbol{PK}_1^*,\sigma_1^*,m_1^*)$ 
	and $(\boldsymbol{PK}_2^*,\sigma_2^*,m_2^*)$ with the following conditions:}
	\begin{itemize}
		\item $\boldsymbol{PK}_i^*\subseteq \boldsymbol{P}$ and $\left| \boldsymbol{PK}_i^* \right|=n$, $i\in \{1,2\}$,
		\item $|(\boldsymbol{PK}_1^*\cup\boldsymbol{PK}_2^*) \cap \boldsymbol{F}|\le 2t-1$,
		\item $(\boldsymbol{PK}_i^*,*,m_i^*) \notin \boldsymbol{Q}$, $i\in \{1,2\}$,
		\item \parbox[t]{0.75\linewidth}{$\mathbf{Verify}(pp,\boldsymbol{PK}_i^*,\sigma_i^*,t,m_i^*)=1$, $i\in \{1,2\}$,}
		\item $\mathbf{Link}(\sigma_1^*,\sigma_2^*)=0$.
	\end{itemize}
	\parbox[t]{0.8\linewidth}{When these hold, the experiment returns 1; otherwise, it returns 0.}
	}
\end{algorithm}

\subsection{Security Models}
\begin{mydefinition}[Witness Extractability]
	We say that a linkable $(t,n)$-threshold ring adaptor signature scheme $\Pi_{\mathcal{R},\Sigma}$ exhibits witness extractability whenever, for a PPT adversary $\mathcal{A}$ in the witness extraction experiment $\mathbf{aWitExt}_{\mathcal{A},\Pi_{\mathcal{R},\Sigma}}$ (specified in Algorithm~\ref{alg:WitExt}), the following is true:
\end{mydefinition}
\[
	\text{Pr}[\mathbf{aWitExt}_{\mathcal{A},\Pi _{\mathcal{R},\Sigma}}(\lambda)=1]\le \epsilon(\lambda).
\]

\begin{mydefinition}[aEUF-CMA Security]
    A linkable $(t,n)$-threshold ring adaptor signature scheme $\Pi_{\mathcal{R},\Sigma}$ is existentially unforgeable under chosen-message attacks for adaptor signatures (aEUF-CMA secure) with respect to insider corruption, if for each PPT algorithm $\mathcal{A}$ that plays the role of the adversary in the forgeability experiment $\mathbf{aSignForge}_{\mathcal{A},\Pi_{\mathcal{R},\Sigma}}$ (specified in Algorithm~\ref{alg:SignForge}), the following condition is satisfied:
\end{mydefinition}
\[
	\text{Pr}[\mathbf{aSignForge}_{\mathcal{A},\Pi _{\mathcal{R},\Sigma}}(\lambda)=1]\le \epsilon(\lambda).
\]

\begin{mydefinition}[Pre-signature Anonymity]
\sloppy
Pre-signature anonymity for a linkable $(t,n)$-threshold ring adaptor signature scheme $\Pi_{\mathcal{R},\Sigma}$ is ensured by the following security condition. For all PPT adversaries $\mathcal{A}$ participating in the $\mathbf{aAnon}_{\mathcal{A},\Pi_{\mathcal{R},\Sigma}}$ experiment (specified in Algorithm~\ref{alg:Anon}), it must hold that:
\end{mydefinition}
\[
	\left| \text{Pr}\left[\mathbf{aAnon}_{\mathcal{A},\Pi _{\mathcal{R},\Sigma}}
	\left(\lambda  \right) =1  \right] -\frac{1}{2} \right|\le \epsilon\left(\lambda \right).
\]

\begin{mydefinition}[Linkability]
    A linkable $(t,n)$-threshold ring adaptor signature scheme $\Pi_{\mathcal{R},\Sigma}$ satisfies the linkability property within the security model accounting for insider corruption, if for each PPT algorithm $\mathcal{A}$ when interacting with the linkability experiment $\mathbf{aLink}_{\mathcal{A},\Pi_{\mathcal{R},\Sigma}}$ (specified in Algorithm~\ref{alg:Link}), the following condition is met:
\end{mydefinition}
\[
	\text{Pr}\left[\mathbf{aLink}_{\mathcal{A},\Pi _{\mathcal{R},\Sigma}}
	\left(\lambda  \right) =1  \right] \le \epsilon\left(\lambda \right). 
\]

\subsection{Construction}
The proposed linkable $(t,n)$-threshold ring adaptor signature scheme $\Pi_{\mathcal{R},\Sigma}$ is constructed as follows. Throughout this paper, for brevity, we will
write $y_i = \prod_{k=i}^{i+t-1} pk_{k}^d$ as a shorthand 
for $y_i = \prod_{k=i (mod\ n)}^{(i+t-1) (mod\ n)} pk_{k}^d$.

\begin{itemize}
	\item $\mathbf{Setup}(1^\lambda) \rightarrow pp$: Given the security parameter $\lambda$, the group generation algorithm $\mathcal{G}(\lambda)$ is invoked, producing $\{p,\mathbb{G}_p,g,h\}$. This results in a cyclic group $\mathbb{G}_p$ of prime order $p$ with two generators, $g$ and $h$. A collision-resistant hash function $H: \{0,1\}^* \to \mathbb{Z}_p$ is also selected. The algorithm then returns the public parameters $pp = \{p, \mathbb{G}_p, g, h, H\}$.
	\item $\mathbf{KeyGen}(pp) \rightarrow (pk, sk)$: On input $pp$, the algorithm first chooses a secret key $sk \in \mathbb{Z}_p^*$ at random. Next, it sets the public key to $pk = g^{sk}$. Finally, the key pair $(pk, sk)$ is output.
	\item $\mathbf{GenR}(pp) \rightarrow (\boldsymbol{W}, w)$: Given public parameters $pp$, the algorithm selects a witness $w \in \mathbb{Z}_p^*$ uniformly at random, computes $W_1=g^{w}$ and $W_2=h^{w}$, sets $\boldsymbol{W}=(W_1,W_2)$, and outputs the statement-witness pair $(\boldsymbol{W}, w)$.
	\item $\mathbf{PreSign}(pp, \boldsymbol{PK}, \boldsymbol{SK},m,\boldsymbol{W}) \rightarrow \tilde{\sigma}$: Given public parameters $pp$, a set of public keys 
	$\boldsymbol{PK}=\left\{pk_0,pk_1,\right.$ $\left.\\\dots,pk_{n-1}\right\}$, a set of secret keys 
	$\boldsymbol{SK}=\left\{sk_j,\dots,\right.$ $\left.\\sk_{j+t-1}\right\}$ corresponding to a subset of $\boldsymbol{PK}$, where $j \in [0,n-t]$ selected by the signer, 
	a message $m\in \{0,1\}^*$, and a statement $\boldsymbol{W}=(W_1,W_2)$ from the hard relation, the algorithm first computes the link tags for every secret key in $\boldsymbol{SK}$ as $tag_i=h^{sk_{j+i}}$ for $i \in [0,t-1]$, and sets $\boldsymbol{TAG} = (tag_0, \dots, tag_{t-1})$. Then, it computes $d=H(\boldsymbol{PK})$, and for every $i \in [0,n-1]$, it computes $y_i=\prod_{k=i}^{i+t-1}{pk_{k}^d}$, and $l=\prod_{k=0}^{t-1}{tag_{k}^d}$. Next, it chooses a random value $r \in \mathbb{Z}_p^*$, and for every $i \in [0,n-1]$ where $i \ne j$, it chooses a random challenge value $c_i \in \mathbb{Z}_p^*$. Then, it computes $R=g^r\cdot W_1\cdot \prod_{i=0,i\ne j}^{n-1}{y_{i}^{c_i}}$ and $T=h^r\cdot W_2\cdot \prod_{i=0,i\ne j}^{n-1}{l^{c_i}}$. After that, it computes the overall challenge $c=H(\boldsymbol{PK},R,T,m)$, and computes $c_j=c-\sum_{i=0,i\ne j}^{n-1}{c_i}$. Finally, it computes the pre-signature component $\tilde{z}=r-c_j\cdot d\cdot \sum_{k=j}^{j+t-1}{ sk_k}$, sets the pre-signature $\tilde{\sigma} = (\tilde{z},\boldsymbol{C},\boldsymbol{TAG} )$, where $\boldsymbol{C}=(c_0,\dots,c_j,\dots,c_{n-1})$, and outputs the pre-signature $\tilde{\sigma}$.
	\item $\mathbf{PreVerify}(pp,\boldsymbol{PK},\tilde{\sigma},t,m,\boldsymbol{W}) \rightarrow 0/1$: Given public parameters $pp$, a set of public keys 
	$\boldsymbol{PK}=\left\{pk_0,pk_1,\right.$ $\left.\\\dots,pk_{n-1}\right\}$, a pre-signature $\tilde{\sigma}=(\tilde{z},\boldsymbol{C},\boldsymbol{TAG})$, 
	a message $m\in \{0,1\}^*$, a threshold $t\in [1,n]$, and a statement $\boldsymbol{W}=(W_1,W_2)$ from the hard relation, 
	the algorithm first computes $d=H(\boldsymbol{PK})$, and for every $i \in [0,n-1]$, it computes $y_i=\prod_{k=i}^{i+t-1}{pk_{k}^{d}}$. Then, it computes $l=\prod_{k=0}^{t-1}{tag_{k}^{d}}$, $R'=g^{\tilde{z}}\cdot W_1\cdot \prod_{i=0}^{n-1}{y_{i}^{c_i}}$, and $T'=h^{\tilde{z}}\cdot W_2\cdot \prod_{i=0}^{n-1}{l^{c_i}}$. Next, it computes $c'=\sum_{i=0}^{n-1}{c_i}$. Finally, the algorithm returns 1, if $c' = H(\boldsymbol{PK},R',T',m)$, and 0 otherwise.
	\item $\mathbf{Adapt}(pp,\tilde{\sigma},w) \rightarrow \sigma$: Taking the public parameters $pp$, a pre-signature 
	$\tilde{\sigma}=\left(\tilde{z},\boldsymbol{C},\right.$ $\left.\\\boldsymbol{TAG}\right)$ and a witness $w \in \mathbb{Z}_p^*$, 
	the algorithm computes $z=\tilde{z}+w$, sets $\sigma=(z,\boldsymbol{C},\boldsymbol{TAG})$, and outputs the full signature $\sigma$.
	\item $\mathbf{Verify}(pp,\boldsymbol{PK},\sigma,t,m) \rightarrow 0/1$: Given public parameters $pp$, a set of public keys 
	$\boldsymbol{PK}=\{pk_0,\dots,pk_{n-1}\}$, a signature $\sigma=(z,\boldsymbol{C},\boldsymbol{TAG})$, 
	a message $m\in \{0,1\}^*$, and a threshold $t\in [1,n]$, the algorithm first computes $d=H(\boldsymbol{PK})$, and for every $i \in [0,n-1]$, it computes $y_i=\prod_{k=i}^{i+t-1}{pk_{k}^{d}}$. Then, it computes $l=\prod_{k=0}^{t-1}{tag_{k}^{d}}$, $R''=g^{z}\cdot \prod_{i=0}^{n-1}{y_{i}^{c_i}}$, and $T''=h^z \cdot \prod_{i=0}^{n-1}{l^{c_i}}$. Next, it computes $c''=\sum_{i=0}^{n-1}{c_i}$. If $c'' = H(\boldsymbol{PK},R'',T'',m)$, it outputs 1. Otherwise, it outputs 0.
	\item $\mathbf{Ext}(pp,\boldsymbol{W},\tilde{\sigma},\sigma) \rightarrow w/\bot$: Given public parameters $pp$, a statement 
	$\boldsymbol{W}=(W_1,W_2)$ from the hard relation, a pre-signature 
	$\tilde{\sigma}=(\tilde{z},\boldsymbol{C},\boldsymbol{TAG})$, and a signature 
	$\sigma=(z,\boldsymbol{C},\boldsymbol{TAG})$, the algorithm computes $w=z-\tilde{z}$. If 
	$W_1=g^w$ and $W_2=h^w$, it outputs the witness $w$. Otherwise, it outputs $\bot$.
	\item $\mathbf{Link}(\sigma',\sigma'') \rightarrow 0/1$: Given two signatures $\sigma'=$ $(z',\boldsymbol{C}',\boldsymbol{TAG}')$ and 
	$\sigma''=(z'',\boldsymbol{C}'',\boldsymbol{TAG}'')$, where 
	$\boldsymbol{TAG}'=\left(tag_0',\dots, tag_{t-1}'\right)$ and 
	$\boldsymbol{TAG}''=\left(tag_0'',\right.$ $\left.\\ \dots, tag_{t-1}''\right)$, the algorithm checks whether there exist $tag_i' \in \boldsymbol{TAG'}$ and $tag_j'' \in \boldsymbol{TAG''}$ such that $tag_i'=tag_j''$, where $i,j \in [0,t-1]$. If such $tag_i'$ and $tag_j''$ exist, it returns 1. Otherwise, it returns 0.
\end{itemize}

\begin{mytheorem}{\label{thm:pre-signature adaptability}}
	The proposed linkable $(t,n)$-threshold ring adaptor signature scheme $\Pi_{\mathcal{R},\Sigma}$ 
	guarantees pre-signature adaptability.
\end{mytheorem}
\begin{proof}
Given public parameters $pp=\left(p,\mathbb{G},g,h,H\right)$, public keys 
$\boldsymbol{PK}=\{pk_0,\dots,pk_{n-1}\}$, message $m\in \{0,1\}^*$, threshold $t\in [1,n]$, 
and pre-signature $\tilde{\sigma}=(\tilde{z},\boldsymbol{C},\boldsymbol{TAG})$. We 
define $\boldsymbol{W}=\{W_1=g^w,W_2=h^w\}$, and signature 
$\sigma=(z=\tilde{z}+w,\boldsymbol{C},\boldsymbol{TAG})$, where $(W_1,w)\in \mathcal{R}$ and 
$(W_2,w)\in \mathcal{R}$, and $w\in \mathbb{Z}_p^*$. Assuming that 
$\mathbf{PreVerify}(pp,\boldsymbol{PK},\tilde{\sigma},t,m,\boldsymbol{W})=1$, we have
\begin{align*}
	c'&=H(\boldsymbol{PK},R',T',m)\\
	&= H\left(\boldsymbol{PK},g^{\tilde{z}}\cdot W_1 \cdot \prod_{i=0}^{n-1}{y_i^{c_i}},h^{\tilde{z}}\cdot W_2 \cdot \prod_{i=0}^{n-1}{l^{c_i}},m\right)\\
	&= H\left(\boldsymbol{PK},g^{\tilde{z}}\cdot g^{w} \cdot \prod_{i=0}^{n-1}{y_i^{c_i}},h^{\tilde{z}}\cdot h^{w} \cdot \prod_{i=0}^{n-1}{l^{c_i}},m\right)\\
	&= H\left(\boldsymbol{PK},g^{z}\cdot \prod_{i=0}^{n-1}{y_i^{c_i}},h^{z} \cdot \prod_{i=0}^{n-1}{l^{c_i}},m\right)\\
	&= H\left(\boldsymbol{PK},R'',T'',m\right)
\end{align*}
which implies that $\mathbf{Verify}(pp,\boldsymbol{PK},\sigma,t,m)=1$.
\end{proof}

\begin{mytheorem}
	The proposed linkable $(t,n)$-threshold ring adaptor signature scheme $\Pi_{\mathcal{R},\Sigma}$ 
	satisfies correctness.
\end{mytheorem}

\begin{proof}
Given public parameters $pp=\left(p,\mathbb{G},g,h,H\right)$, public keys 
$\boldsymbol{PK}=\{pk_0,\dots,pk_{n-1}\}$, secret keys $\boldsymbol{SK}=\{sk_j,\dots,sk_{j+t-1}\}$, 
message $m\in \{0,1\}^*$, threshold $t\in [1,n]$, 
and $\boldsymbol{W}=\{W_1=g^w,W_2=h^w\}$, where $pk_j=g^{sk_j},\dots,pk_{j+t-1}=g^{sk_{j+t-1}}$, 
$(W_1,w)\in \mathcal{R}$, $(W_2,w)\in \mathcal{R}$, and $w\in \mathbb{Z}_p^*$. 
For the pre-signature $\tilde{\sigma}=(\tilde{z},\boldsymbol{C},\boldsymbol{TAG})
\leftarrow \mathbf{PreSign}(pp,\boldsymbol{PK},\boldsymbol{SK},m,\boldsymbol{W})$, 
it holds that $\tilde{z}=r-c_j\cdot d\cdot \sum_{k=j}^{j+t-1}{sk_k}$ 
for some $r\in \mathbb{Z}_p^*$. Since
\begin{align*}
	R'&=g^{\tilde{z}}\cdot W_1 \cdot \prod_{i=0}^{n-1}{y_i^{c_i}}\\
	&=g^{r-c_j\cdot d\cdot \sum_{k=j}^{j+t-1}{sk_k}}\cdot W_1 \cdot y_j^{c_j} \prod_{i=0,i\ne j}^{n-1}{y_i^{c_i}}\\
	&=g^{r-c_j\cdot d\cdot \sum_{k=j}^{j+t-1}{sk_k}}\cdot W_1\cdot \left(\prod_{k=j}^{j+t-1}{pk_k^{d}}\right)^{c_j}   \prod_{i=0,i\ne j}^{n-1}{y_i^{c_i}}\\
	&=g^{r-c_j\cdot d\cdot \sum_{k=j}^{j+t-1}{sk_k}}\cdot W_1\cdot g^{c_j\cdot d\cdot \sum_{k=j}^{j+t-1}{sk_k}}   \prod_{i=0,i\ne j}^{n-1}{y_i^{c_i}}\\
	&=g^{r}\cdot W_1 \cdot \prod_{i=0,i\ne j}^{n-1}{y_i^{c_i}}\\
	&=R,
\end{align*}
and
\begin{align*}
	T'&=h^{\tilde{z}}\cdot W_2 \cdot \prod_{i=0}^{n-1}{l^{c_i}}\\
	&=h^{r-c_j\cdot d\cdot \sum_{k=j}^{j+t-1}{sk_k}}\cdot W_2 \cdot l^{c_j}  \prod_{i=0,i\ne j}^{n-1}{l^{c_i}}\\
	&=h^{r-c_j\cdot d\cdot \sum_{k=j}^{j+t-1}{sk_k}}\cdot W_2 \cdot \left(\prod_{k=0}^{t-1}{tag_k^d}\right)^{c_j} \prod_{i=0,i\ne j}^{n-1}{l^{c_i}}\\
	&=h^{r-c_j\cdot d\cdot \sum_{k=j}^{j+t-1}{sk_k}}\cdot W_2  \cdot h^{c_j\cdot d\cdot \sum_{k=j}^{j+t-1}{sk_k}}  \prod_{i=0,i\ne j}^{n-1}{l^{c_i}}\\
	&=h^{r}\cdot W_2 \cdot \prod_{i=0,i\ne j}^{n-1}{l^{c_i}}\\
	&=T,
\end{align*}
we have $H(\boldsymbol{PK},R',T',m)=H\left(\boldsymbol{PK},R,T,m\right)=c=\sum_{i=0}^{n-1}{c_i}=c'$, then 
$\mathbf{PreVerify}\left(pp,\boldsymbol{PK},\tilde{\sigma},t,m,\boldsymbol{W}\right)$ outputs 1. 
By Theorem~\ref{thm:pre-signature adaptability}, we have
$\mathbf{Verify}\left(pp,\boldsymbol{PK},\sigma,t,m\right)=1$ for 
$\sigma=\left(z=\tilde{z}+w,\right.$ $\left.\\\boldsymbol{C},\boldsymbol{TAG}\right)\leftarrow \mathbf{Adapt}(pp,\tilde{\sigma},w)$, 
which completes the proof.
\end{proof}

\begin{mytheorem}
	The proposed linkable $(t,n)$-threshold ring adaptor signature scheme $\Pi_{\mathcal{R},\Sigma}$ 
	guarantees witness extractability.
\end{mytheorem}

\begin{proof}
If a PPT adversary $\mathcal{A}$ can win the witness extractability experiment $\mathbf{aWitExt}_{\mathcal{A},\Pi_{\mathcal{R},\Sigma}}$ with non-negligible probability, this adversary could then be used by a simulator $\mathcal{B}$ to break the collision resistance of the hash function. The construction of $\mathcal{B}$ is as follows.
Given the system parameter $pp=(p,\mathbb{G},g,h,H)$, $\mathcal{B}$ will outupt a pair of hash collisions for $H$.

 $\mathsf{Setup}$. $\mathcal{B}$ initializes the signature query set $\boldsymbol{Q}=\emptyset$, 
  and sends $pp=(p,\mathbb{G},g,h,H)$ to $\mathcal{A}$. $\mathcal{A}$ runs $\mathbf{GenR}(pp) \to (\boldsymbol{W},w)$, 
  and sends $\boldsymbol{W}$ to $\mathcal{B}$.

 $\mathsf{KeyGen}$. $\mathcal{B}$ runs $\mathbf{KeyGen}(pp) \to (pk_i,sk_i)$, 
  for $i \in [0,n-1]$, sets $\boldsymbol{PK}=(pk_0,pk_1,\dots,pk_{n-1})$, 
  and sends $\boldsymbol{PK}$ to $\mathcal{A}$.

 $\mathsf{Query}$. The procedure for $\mathcal{B}$ to answer the adversary's oracle queries is detailed below.
  \begin{itemize}[]
  	\item $\mathcal{SO}(\boldsymbol{PK},\boldsymbol{S},m)$. 
	$\mathcal{B}$ runs $\mathbf{Sign}\left(pp,\boldsymbol{PK},\left\{ sk_i \right\} _{i\in \boldsymbol{S}},m\right)$ $\to \sigma$, 
	sets $\boldsymbol{Q}=\boldsymbol{Q} \cup \left\{m\right\}$, 
	and returns the signature $\sigma$.
	\item $\mathcal{PSO}(\boldsymbol{PK},\boldsymbol{S},m,\boldsymbol{W})$. 
	$\mathcal{B}$ runs $\mathbf{PreSign}\left(pp,\boldsymbol{PK},\right.$ $\left.\\\left\{ sk_i \right\} _{i\in \boldsymbol{S}},m,\boldsymbol{W}\right) \to \tilde{\sigma}$, 
	sets $\boldsymbol{Q}=\boldsymbol{Q} \cup \left\{m \right\}$, 
	and returns the pre-signature $\tilde{\sigma}$.
  \end{itemize}

 $\mathsf{Witness\ Extraction}$. $\mathcal{A}$ sends the target message $\left(\boldsymbol{S}^*,\right.$ $\left.\\m^*,\boldsymbol{W}^*=(W_1^*,W_2^*)\right)$ 
	to $\mathcal{B}$. The simulator $\mathcal{B}$ runs 
	$\mathbf{PreSign}\left(pp,\boldsymbol{PK},\{sk_i\} _{i\in \boldsymbol{S}^*},m^*,\boldsymbol{W}^*\right)\to \tilde{\sigma}^*$, 
	where $\tilde{\sigma}^*$ $=(\tilde{z}^*,\boldsymbol{C}^*,\boldsymbol{TAG}^*)$ then sends $\tilde{\sigma}^*$ to $\mathcal{A}$. 
	$\mathcal{A}$ generates a signature $\sigma^*$ of $m^*$, where $\sigma^*=(z^*,\boldsymbol{C}^*,\boldsymbol{TAG}^*)$, 
	and sends $\sigma^*$ to $\mathcal{C}$. 
	$\mathcal{C}$ runs $\mathbf{Ext}(pp,\boldsymbol{W^*},\tilde{\sigma}^*,\sigma^*)\to w^*$, where $w^*=z^*-\tilde{z}^*$. 
	At the end $\mathcal{A}$ wins the experiment with the following conditions:
	\begin{itemize}[]
		\item $m^* \notin \boldsymbol{Q}$,
		\item $\mathbf{Verify}(pp,\boldsymbol{PK},\sigma^*,t,m^*)=1$,
		\item $(W_1^*,w^*)\notin \mathcal{R} \lor (W_2^*,w^*)\notin \mathcal{R}$.
	\end{itemize}

	As described in the $\mathbf{PreSign}$ algorithm, we have $R=g^r\cdot W_1\cdot \prod_{i=0,i\ne j}^{n-1}{y_{i}^{c_i}}$, 
	$T=h^r\cdot W_2\cdot \prod_{i=0,i\ne j}^{n-1}{l^{c_i}}$, $c=H(\boldsymbol{PK},R,T,m^*)$, $c_j=c-\sum_{i=0,i\ne j}^{n-1}{c_i}$, 
	and $\tilde{z}^*=r-c_j\cdot d\cdot \sum_{k=j}^{j+t-1}{sk_k}$. 
	Since $\sigma^*$ is a valid signature, we have $R^*=g^{z^*}\cdot \prod_{i=0}^{n-1}{y_{i}^{c_i}}$, 
	$T^*=h^{z^*}\cdot \prod_{i=0}^{n-1}{l^{c_i}}$, and $\sum_{i=0}^{n-1}{c_i}=c^*=H(\boldsymbol{PK},R^*,T^*,m^*)$. Therefore,
	\begin{align*}
		R^*&=g^{z^*}\cdot \prod_{i=0}^{n-1}{y_{i}^{c_i}}\\
		&=g^{\tilde{z}^*} g^{w^*}\cdot \prod_{i=0}^{n-1}{y_{i}^{c_i}}\\
		&=g^{r-c_j\cdot d\cdot \sum_{k=j}^{j+t-1}{sk_k}} g^{w^*}\cdot y_j^{c_j}  \prod_{i=0,i\ne j}^{n-1}{y_{i}^{c_i}}\\
		&=g^{r-c_j\cdot d\cdot \sum_{k=j}^{j+t-1}{sk_k}} g^{w^*} \cdot \left(\prod_{k=j}^{j+t-1}{pk_k^d}\right)^{c_j} \prod_{i=0,i\ne j}^{n-1}{y_{i}^{c_i}}\\
		&=g^{r-c_j\cdot d\cdot \sum_{k=j}^{j+t-1}{sk_k}} g^{w^*} \cdot g^{c_j\cdot d\cdot \sum_{k=j}^{j+t-1}{sk_k}} \prod_{i=0,i\ne j}^{n-1}{y_{i}^{c_i}}\\
		&=g^r g^{w^*}\cdot \prod_{i=0,i\ne j}^{n-1}{y_{i}^{c_i}},
	\end{align*}
	and 
	\begin{align*}
		T^*&=h^{z^*}\cdot \prod_{i=0}^{n-1}{l^{c_i}}\\
		&=h^{\tilde{z}^*}\cdot h^{w^*}\cdot \prod_{i=0}^{n-1}{l^{c_i}}\\
		&=h^{r-c_j\cdot d\cdot \sum_{k=j}^{j+t-1}{sk_k}}h^{w^*}\cdot l^{c_j}  \prod_{i=0,i\ne j}^{n-1}{l^{c_i}}\\
		&=h^{r-c_j\cdot d\cdot \sum_{k=j}^{j+t-1}{sk_k}} h^{w^*} \cdot \left(\prod_{k=0}^{t-1}{tag_k^{d}}\right)^{c_j}  \prod_{i=0,i\ne j}^{n-1}{l^{c_i}}\\
		&=h^{r-c_j\cdot d\cdot \sum_{k=j}^{j+t-1}{sk_k}} h^{w^*} \cdot h^{c_j\cdot d\cdot \sum_{k=j}^{j+t-1}{sk_k}}  \prod_{i=0,i\ne j}^{n-1}{l^{c_i}}\\
		&=h^r\ h^{w^*}\cdot \prod_{i=0,i\ne j}^{n-1}{l^{c_i}}.
	\end{align*}
	Since $(W_1^*,w^*)\notin \mathcal{R} \lor (W_2^*,w^*)\notin \mathcal{R}$, we have $W_1^*\ne g^{w^*}$ or $W_2^*\ne h^{w^*}$.
	Thus, we have $R^*\ne R$ or $T^*\ne T$, then the hash inputs $(\boldsymbol{PK},R^*,T^*,m^*)$ and $(\boldsymbol{PK},R,T,m^*)$ are different.
	However, since $H(\boldsymbol{PK},R^*,T^*,m^*)=c^*=\sum_{i=0}^{n-1}{c_i}=c=H(\boldsymbol{PK},R,T,m)$, we obtain a pair of hash collisions for $H$:
	$(\boldsymbol{PK},R^*,T^*,m^*)$ and $(\boldsymbol{PK},R,T,m^*)$.
\end{proof}

\begin{mytheorem}{\label{thm:unforgeability}}
	The proposed linkable $(t,n)$-threshold ring adaptor signature scheme $\Pi_{\mathcal{R},\Sigma}$ provides $(\epsilon(\lambda),q_h,q_s,q_{ps})$-unforgeability within the aEUF-CMA security model under insider corruption, provided that the $\epsilon'(\lambda)$-DL assumption holds in the group $\mathbb{G}_p$ with $p > q_{ps}(q_h+q_{ps}-1)+q_{s}(q_h+q_{s}-1)$, and provided that $\epsilon'(\lambda) \ge \frac{n-t+1}{32nq_k}\epsilon(\lambda)$. Here, $q_h$, $q_s$, $q_{ps}$, $q_k$, $n$, and $t$ denote the number of hash queries, signature queries, pre-signature queries, public keys, ring size, and threshold, respectively.
\end{mytheorem}

\begin{proof}
If a PPT adversary $\mathcal{A}$ can break the unforgeability of $\Pi_{\mathcal{R},\Sigma}$ in the aEUF-CMA model with insider corruption, then one is able to construct a simulator $\mathcal{B}$ that, using $\mathcal{A}$, can in turn violate the DL assumption. The construction of $\mathcal{B}$ proceeds as follows. Given the system parameter $pp=(p,\mathbb{G},g,h,H)$, a statement $\boldsymbol{W}$ of the hard relation, and 
a DL problem instance $(g,g^\alpha)$, $\mathcal{B}$ 
will outupt $\alpha$ with non-negligible probability.

 $\mathsf{Setup}$. $\mathcal{B}$ initializes the corrupted party set 
  $\boldsymbol{F}=\emptyset$ and the signature query set $\boldsymbol{Q}=\emptyset$,
  and sends $pp=(p,\mathbb{G},g,h)$ to $\mathcal{A}$.

   $\mathsf{KeyGen}$. $\mathcal{B}$ picks a random index $i^*\in [1,q_k]$. 
  $\mathcal{B}$ runs $\mathbf{KeyGen}(pp) \to (pk_i,sk_i)$, 
  for $i \in [1,q_k],i\ne i^*$, sets $pk_{i^*}=g^\alpha$. 
  By doing this, $\mathcal{B}$ implicitly defines $sk_{i^*}=\alpha$.
  $\mathcal{B}$ sets $\boldsymbol{P}=(pk_1,pk_2,\dots,pk_{q_k})$, 
  and sends $\boldsymbol{P}$ and $\boldsymbol{W}$ to $\mathcal{A}$.

 $\mathsf{Query}$. The procedure for $\mathcal{B}$ to answer the adversary's oracle queries is detailed below.
  \begin{itemize}[]
	\item $\mathcal{HO}(\{0,1\}^*)$. $\mathcal{B}$ simulates the hash function $H$
	as a random oralce.
  	\item $\mathcal{CO}(i)$. If $i=i^*$, $\mathcal{B}$ declares failure and aborts,  
  	otherwise, $\mathcal{B}$ responds private key $sk_i$.
  	\item $\mathcal{PSO}(\boldsymbol{PK},\boldsymbol{S},m,\boldsymbol{W})$. 
  	The algorithm $\mathcal{B}$ first checks the set $\boldsymbol{S}$. Should the indices within $\boldsymbol{S}$ be non-contiguous, it outputs $\bot$.
  	In the case where $i^* \notin \boldsymbol{S}$, $\mathcal{B}$ executes the $\mathbf{PreSign}$ algorithm faithfully. 
	Conversely, if $i^* \in \boldsymbol{S}$, $\mathcal{B}$ generates the pre-signature with a modification: the procedure is identical to $\mathbf{PreSign}$ except for the following. $\mathcal{B}$ selects random elements: $c_i \in \mathbb{Z}_p$ for each $i \in [0,n-1]$, $tag_j \in \mathbb{G}_p$ for each $j \in [0,|\boldsymbol{S}|-1]$, and $\tilde{z} \in \mathbb{Z}_p$. It then computes $R=g^{\tilde{z}}\cdot W_1 \cdot
	\prod_{i=0}^{n-1}{y_{i}^{c_i}}$, $T=h^{\tilde{z}}\cdot W_2 \cdot \prod_{i=0}^{n-1}{l^{c_i}}$. 
	Subsequently, $\mathcal{B}$ programs the random oracle $\mathcal{HO}$ so that the challenge value satisfies $c=H(\boldsymbol{PK},R,T,m) = \sum_{i=0}^{n-1}{c_i}$.
	If this specific value for $H$ has already been defined in $\mathcal{HO}$, $\mathcal{B}$ aborts the simulation due to failure.
	Finally, $\mathcal{B}$ updates the query set as $\boldsymbol{Q} \leftarrow \boldsymbol{Q} \cup \{(\boldsymbol{PK},\boldsymbol{S},m)\}$ and sends the pre-signature tuple $\tilde{\sigma}=(\tilde{z},\boldsymbol{C},\boldsymbol{TAG})$ to $\mathcal{A}$, where $\boldsymbol{C}=(c_0,\dots,c_{n-1})$ and $\boldsymbol{TAG}=(tag_0,\dots,tag_{|\boldsymbol{S}|-1})$.
	\item $\mathcal{SO}(\boldsymbol{PK},\boldsymbol{S},m)$. 
  	The algorithm $\mathcal{B}$ first checks the set $\boldsymbol{S}$. Should the indices within $\boldsymbol{S}$ be non-contiguous, it outputs $\bot$.
  	In the case where $i^* \notin \boldsymbol{S}$, $\mathcal{B}$ executes the $\mathbf{Sign}$ algorithm faithfully. 
	Conversely, if $i^* \in \boldsymbol{S}$, $\mathcal{B}$ generates the signature with a modification: the procedure is identical to $\mathbf{Sign}$ except for the following. $\mathcal{B}$ selects random elements: $c_i\in \mathbb{Z}_p$ for each $i\in [0,n-1]$,
	$tag_j\in \mathbb{G}_p$ for each $j\in [0,|\boldsymbol{S}|-1]$, and $z\in \mathbb{Z}_p$. It then computes $R=g^{z}\cdot \prod_{i=0}^{n-1}{y_{i}^{c_i}}$, $T=h^{z}\cdot \prod_{i=0}^{n-1}{l^{c_i}}$.
	Subsequently, $\mathcal{B}$ programs the random oracle $\mathcal{HO}$ so that the challenge value satisfies $c=H(\boldsymbol{PK},R,T,m) = \sum_{i=0}^{n-1}{c_i}$.
	If this specific value for $H$ has already been defined in $\mathcal{HO}$, $\mathcal{B}$ aborts the simulation due to failure.
	Finally, $\mathcal{B}$ updates the query set as $\boldsymbol{Q} \leftarrow \boldsymbol{Q} \cup \{(\boldsymbol{PK},\boldsymbol{S},m)\}$ and sends the signature tuple $\sigma=(z,\boldsymbol{C},\boldsymbol{TAG})$ to $\mathcal{A}$, where $\boldsymbol{C}=(c_0,\dots,c_{n-1})$ and $\boldsymbol{TAG}=(tag_0,\dots,tag_{|\boldsymbol{S}|-1})$.
  \end{itemize}

 $\mathsf{Forgery}$. $\mathcal{A}$ sends the target message
  $\left(\boldsymbol{PK}^*,\boldsymbol{S}^*,m^*\right)$ to $\mathcal{B}$. $\mathcal{B}$ 
  runs $\mathcal{PSO}(\boldsymbol{PK}^*,\boldsymbol{S}^*,m^*,\boldsymbol{W})\to \tilde{\sigma}^*$, 
  then sends $\tilde{\sigma}^*$ to $\mathcal{A}$. At the end, 
  $\mathcal{A}$ forges a signature $\sigma^*=(z^*,\boldsymbol{C}^*=\{c_0^*,\dots,c_{n-1}^*\},*)$ 
  on $(\boldsymbol{PK}^*,m^*)$ with the following conditions:
  \begin{itemize}[]
	\item $\boldsymbol{PK}^*\subseteq \boldsymbol{P}$ and
	$\left| \boldsymbol{PK}^* \right|=n$,
	\item $|\boldsymbol{PK}^*\cap \boldsymbol{F}|\le t-1$,
	\item $(\boldsymbol{PK}^*,*,m^*) \notin \boldsymbol{Q}$,
	\item $\mathbf{Verify}(pp,\boldsymbol{PK}^*,\sigma^*,t^*,m^*)=1$.
  \end{itemize}
  If $pk_{i^*}\notin \boldsymbol{PK}^*$, $\mathcal{B}$ declares failure and aborts. 
  Otherwise, $\mathcal{B}$ computes $\alpha$ as follows. 

   $\mathcal{B}$ computes $R^*=g^{z^*}\cdot \prod_{i=0}^{n-1}{y_{i}^{c_i^*}}$ as described in $\mathbf{Verify}$ algorithm. 
  Subsequently, $\mathcal{B}$ manipulates the random-oracle replay until the tuple $(\boldsymbol{PK}^*,R^*,T^*,m^*)$ is submitted to $H$. At that moment, $\mathcal{B}$ programs the oracle to output a distinct challenge $c'$, leveraging the forking lemma \cite{bellare2006multi}.  
  $\mathcal{A}$ returns another signature $\sigma'=(z',\boldsymbol{C}'=\{c_0',\dots,c_{n-1}'\},*)$. 
  Given the validity of both $\sigma^*$ and $\sigma'$, the following equality holds:  
  \[
  	R^*=g^{z^*}\cdot \prod_{i=0}^{n-1}{y_{i}^{c_i^*}}=g^{z'}\cdot \prod_{i=0}^{n-1}{y_{i}^{c_i'}}.
  \]
  Note that it is impossible to have $c_i^*=c_i'$ for all $i\in [0,n-1]$ (since 
  $\sum_{i=0}^{n-1}{c_i^*}\ne \sum_{i=0}^{n-1}{c_i'}$). If $c_{i^*}^*=c_{i^*}'$, 
  $\mathcal{B}$ declares failure and aborts. With probability at least $1/n$, we have
  $c_{i^*}^*\ne c_{i^*}'$. For ease of reading, we let 
  $\sum_{i}=d\cdot \sum_{k=i}^{i+t-1}{sk_k}$. Therefore, 
  \begin{align*}
	g^{z^*}\cdot \prod_{i=0}^{n-1}{y_{i}^{c_i^*}}
	&= g^{z^*}\cdot \prod_{i=0}^{n-1}{\left( \prod_{k=i}^{i+t-1}{pk_{k}^{d}} \right) ^{c_{i}^{*}}}\\
	&= g^{z^*}\cdot \prod_{i=0}^{n-1}{\left( g^{c_{i}^{*}\cdot \sum_i} \right)}\\
	&= g^{z^*}\cdot g^{\sum_{i=0}^{n-1}{\left( c_{i}^{*}\cdot \sum_i \right)}}\\
	&= g^{z^*+\sum_{i=0,i\ne i^*}^{n-1}{\left( c_{i}^{*}\cdot \sum_i \right)}}\cdot g^{c_{i^*}^{*}\cdot \left( \alpha \cdot d+\sum_{i^*+1}{} \right)}.\\
  \end{align*}
Similarly, we have 
$g^{z'}\cdot \prod_{i=0}^{n-1}{y_{i}^{c_i'}}=g^{z'+\sum_{i=0,i\ne i^*}^{n-1}{\left( c_{i}'\cdot \sum_i \right)}}\cdot g^{c_{i^*}'\cdot \left( \alpha \cdot d+\sum_{i^*+1}{} \right)}$. 
Then $\mathcal{B}$ can get the discrete logarithm $\alpha$ of $g^\alpha$ by computing that: 
\[
  \alpha=\frac{z^*-z'}{(c_{i^*}^*-c_{i^*}')\cdot d}-\sum_{k=i^*+1}^{i^*+t-1}{sk_k},
\]
as a break of the DL assumption. 

 $\mathsf{Probability\ Analysis}$. We assess the probability that the simulation does not abort. 
  \begin{itemize}[]
		\item Considering $q_c$ queries made to the $\mathcal{CO}$ oracle, the success probability in the first query is calculated as $(1-\frac{1}{q_k})$. For the second query, it is at least $(1-\frac{1}{q_k-1})$. Extending this to $q_c$ queries, the cumulative success probability is given by $(1-\frac{1}{q_k})(1-\frac{1}{q_k-1})\cdots(1-\frac{1}{q_k-q_c+1})=1-\frac{q_c}{q_k}$. 
	\item Consider $q_{ps}$ queries submitted to the $\mathcal{PSO}$ oracle. The success probability for the very first such query is $(1-\frac{q_h}{p})$, where $q_h$ counts queries previously made to the $\mathcal{HO}$ oracle. Extending this to all $q_{ps}$ queries, the lower bound for the overall success probability is:
	\[
  \begin{array}{c}
    \left( 1-\frac{q_h}{p} \right)\left( 1-\frac{q_h+1}{p} \right)\dots 
		\left( 1-\frac{q_h+q_{ps}-1}{p} \right) \\ \ge 1-\frac{q_{ps}(q_h+q_{ps}-1)}{p}.
  \end{array}
  \]
	\item A parallel analysis applies to the $\mathcal{SO}$ oracle. For $q_s$ queries, the success probability is no less than $1-\frac{q_{s}(q_h+q_{s}-1)}{p}$.
	\item Suppose $n_a$ is the number of uncorrupted parties in $\boldsymbol{PK}^*$, 
	here $n-t+1\le n_a\le n$. The probability of $pk_{i^*}\notin \boldsymbol{PK}^*$ in the
	forgery phase is $(1-\frac{1}{q_k-q_c})(1-\frac{1}{q_k-q_c-1})\dots 
	(1-\frac{1}{q_k-q_c-n_a+1})=\frac{q_k-q_c-n_a}{q_k-q_c}\le \frac{q_k-q_c-(n-t+1)}{q_k-q_c}$. 
	Then the probability of $pk_{i^*}\in \boldsymbol{PK}^*$ is at least $\frac{n-t+1}{q_k-q_c}$. 
  \end{itemize}
  
 Given that $\mathcal{A}$ forges a signature with probability $\epsilon(\lambda)$, we have that $\mathcal{B}$ avoids abortion prior to rewinding with the following probability:
\begin{align*}
	\epsilon_b(\lambda)& \ge \epsilon(\lambda) \left( 1-\frac{q_c}{q_k} \right) \left( 1-\frac{q_{ps}(q_h+q_{ps}-1)}{p} \right) \\&\quad \left( 1-\frac{q_{s}(q_h+q_{s}-1)}{p} \right) \left(\frac{n-t+1}{q_k-q_c}\right)\\
	&= \epsilon(\lambda) \left(\frac{n-t+1}{q_k} \right) \left( 1-\frac{q_{ps}(q_h+q_{ps}-1)}{p} \right) \\&\quad \left( 1-\frac{q_{s}(q_h+q_{s}-1)}{p} \right).
\end{align*}
	
The generalized forking lemma \cite{bagherzandi2008multisignatures} guarantees a success probability of at least $\frac{\epsilon_b(\lambda)}{8}$ for the rewinding procedure, provided that $p>8nq_h/\epsilon_b(\lambda)$. (The time required for this simulation is $\tau\cdot 8n^2q_h/\epsilon_b(\lambda) \cdot \ln(8n/\epsilon_b(\lambda))$ when $\mathcal{A}$ runs in time $\tau$.) We have  $c_{i^*}^*\ne c_{i^*}'$ with probability 
at least $1/n$, so the probability $\epsilon'(\lambda)$ of $\mathcal{B}$ successfully 
breaking the DL assumption is:
\begin{align*}
   \epsilon'(\lambda)&\ge \frac{\epsilon_b(\lambda)}{8n}\\&= \frac{\epsilon(\lambda)}{8n} \left(\frac{n-t+1}{q_k} \right)  
	 \left( 1-\frac{q_{ps}(q_h+q_{ps}-1)}{p} \right)\\&\quad \left( 1-\frac{q_{s}(q_h+q_{s}-1)}{p} \right),
\end{align*}
if $p>q_{ps}(q_h+q_{ps}-1)$, $p>q_{s}(q_h+q_{s}-1)$, and $p>8nq_h/\epsilon_b(\lambda)$. If we 
take $p>q_{ps}(q_h+q_{ps}-1)+q_{s}(q_h+q_{s}-1)$, the probability $\epsilon'(\lambda)$ can be 
further simplified to $\epsilon'(\lambda) \ge \frac{n-t+1}{32nq_k}\epsilon(\lambda)$.
\end{proof}

\begin{mytheorem}
	We state that the proposed linkable $(t,n)$-threshold ring adaptor signature scheme $\Pi_{\mathcal{R},\Sigma}$ achieves pre-signature anonymity, subject to the constraints $p>q_{ps}(q_h+q_{ps}-1)$ and $p>q_{s}(q_h+q_{s}-1)$. Here, $q_h$, $q_s$, and $q_{ps}$
	are parameters representing the counts of hash queries, signature queries, and pre-signature queries, respectively.
\end{mytheorem}

\begin{proof}
We demonstrate the constructibility of a simulator $\mathcal{B}$ which furnishes perfect anonymity in the random oracle model, based on the system parameters $pp=\left(p,\mathbb{G},g,\right.$ $\left.\\h,H\right)$ and a given statement $\boldsymbol{W}$ of the hard relation.

 $\mathsf{Setup}$. $\mathcal{B}$ initializes the signature query set $\boldsymbol{Q}=\emptyset$, 
  and sends $pp=(p,\mathbb{G},g,h)$ to $\mathcal{A}$.

 $\mathsf{KeyGen}$. $\mathcal{B}$ runs $\mathbf{KeyGen}(pp) \to (pk_i,sk_i)$, 
  for $i \in [1,q_k]$, and sets $\boldsymbol{P}=(pk_1,pk_2,\dots,pk_{q_k})$, 
  and sends $\boldsymbol{P}$ and $\boldsymbol{W}$ to $\mathcal{A}$.

 $\mathsf{Query}$. The procedure for $\mathcal{B}$ to answer the adversary's oracle queries is consistent with the responses to the $\mathcal{HO}$, $\mathcal{PSO}$ and $\mathcal{SO}$ oracles in the proof of Theorem~\ref{thm:unforgeability}.

$\mathsf{Challenge}$. $\mathcal{A}$ sends the target message 
	$\left(\boldsymbol{PK}^*,\boldsymbol{S}_0^*,\boldsymbol{S}_1^*,\right.$ $\left.\\m^*,\boldsymbol{W}^*=(W_1^*,W_2^*)\right)$ to $\mathcal{B}$, where $\boldsymbol{PK}^*\subseteq \boldsymbol{P}$, $|\boldsymbol{PK}^*|=n$, $|\boldsymbol{S}_0^*|=|\boldsymbol{S}_1^*|=t$, 
	$\boldsymbol{S}_0^*\ne \boldsymbol{S}_1^*$, $(*,\boldsymbol{S}_0^*,*) \notin \boldsymbol{Q}$, 
	and $(*,\boldsymbol{S}_1^*,*) \notin \boldsymbol{Q}$. $\mathcal{B}$ generates the pre-signature with a modification: the procedure is identical to $\mathbf{PreSign}$, except for the following. $\mathcal{B}$ selects random $c_i\in \mathbb{Z}_p$ for each $i\in [0,n-1]$, $tag_j\in \mathbb{G}_p$ for each $j\in [0,t-1]$, and $\tilde{z}\in \mathbb{Z}_p$. It then computes $R=g^{\tilde{z}}\cdot W_1 \cdot \prod_{i=0}^{n-1}{y_{i}^{c_i}}$, 
	$T=h^{\tilde{z}}\cdot W_2 \cdot \prod_{i=0}^{n-1}{l^{c_i}}$. 
	Subsequently, $\mathcal{B}$ programs the random oracle $\mathcal{HO}$ so that the challenge value satisfies $c=H(\boldsymbol{PK},R,T,m)=\sum_{i=0}^{n-1}{c_i}$.
	If this specific value for $H$ has already been defined in $\mathcal{HO}$, $\mathcal{B}$ aborts the simulation due to fail. 
	Finally, $\mathcal{B}$ sends $\tilde{\sigma}^*=(\tilde{z},\boldsymbol{C},\boldsymbol{TAG})$ 
	to $\mathcal{A}$, where $\boldsymbol{C}=(c_0,\dots,c_{n-1})$ and 
	$\boldsymbol{TAG}=(tag_0,\dots,tag_{t-1})$. In the end, $\mathcal{B}$ samples a random bit
	$b\in \{0,1\}$. 

 $\mathsf{Output}$. $\mathcal{A}$ returns a bit $b' \in \{0,1\}$. 
	$\boldsymbol{S_0^*}$ and $\boldsymbol{S_1^*}$ is not used in 
	the generation of $\tilde{\sigma}^*$ (i.e., the bit $b$ is not used in generation 
	of $\tilde{\sigma}^*$). Thus, the probability of the adversary's success does not exceed $1/2$.

 $\mathsf{Probability\ Analysis}$. We assess the probability that the simulation does not abort. 
	\begin{itemize}[]
		\item Consider $q_{ps}$ queries submitted to the $\mathcal{PSO}$ oracle. The success probability for the very first such query is $(1-\frac{q_h}{p})$, where $q_h$ counts queries previously made to the $\mathcal{HO}$ oracle. Extending this to all $q_{ps}$ queries, the lower bound for the overall success probability is:
	\[
  \begin{array}{c}
    \left( 1-\frac{q_h}{p} \right)\left( 1-\frac{q_h+1}{p} \right)\dots 
		\left( 1-\frac{q_h+q_{ps}-1}{p} \right) \\ \ge 1-\frac{q_{ps}(q_h+q_{ps}-1)}{p}.
  \end{array}
  \]
	\item A parallel analysis applies to the $\mathcal{SO}$ oracle. For $q_s$ queries, the success probability is no less than $1-\frac{q_{s}(q_h+q_{s}-1)}{p}$.
		\item In the challenge phase, the probability that the 
		input to the $\mathcal{HO}$ has not been queried before is 
		$1-\frac{q_{h}+q_{ps}+q_{s}}{p}$.
	\end{itemize}

 In the above experiment, the probability of $\mathcal{B}$ does not abort 
is at least: $(1-\frac{q_{ps}(q_h+q_{ps}-1)}{p})(1-\frac{q_{s}(q_h+q_{s}-1)}{p})(1-\frac{q_{h}+q_{ps}+q_{s}}{p})$, 
if $p>q_{ps}(q_h+q_{ps}-1)$, and $p>q_{s}(q_h+q_{s}-1)$. If $\mathcal{B}$ does not abort, 
then $\left|\text{Pr}\left[\mathbf{aAnon}_{\mathcal{A},\Pi _{\mathcal{R},\Sigma}}
	\left(\lambda  \right) =1  \right] -\frac{1}{2}\right| \le \epsilon(\lambda)$.
\end{proof}

\begin{mytheorem}
	The linkable $(t,n)$-threshold ring adaptor signature scheme $\Pi_{\mathcal{R},\Sigma}$ guarantees $\left(\epsilon(\lambda),q_h, q_{ps},q_s\right)$-linkability with respect to insider corruption, subject to the constraints that the $\epsilon'(\lambda)$-discrete logarithm assumption holds on $\mathbb{G}_p$, that $p>q_{ps}(q_h+q_{ps}-1)+q_{s}(q_h+q_{s}-1)$, and that $\epsilon'(\lambda) \ge \frac{2n-2t+1}{32nq_k}\epsilon(\lambda)$, where the quantities $q_h$, $q_s$, and $q_{ps}$ represent the numbers of hash, signature, and pre-signature queries.
\end{mytheorem}

\begin{proof}
The existence of a PPT adversary $\mathcal{A}$ winning the linkability experiment $\mathbf{aLink}_{\mathcal{A},\Pi_{\mathcal{R},\Sigma}}$ enables the construction of a simulator $\mathcal{B}$ that uses $\mathcal{A}$ to break the DL assumption. On receiving the system parameters $pp=(p,\mathbb{G},g,h,H)$, a statement $\boldsymbol{W}$, and a DL instance $(g,g^\alpha)$, $\mathcal{B}$ succeeds with non-negligible probability in thereby computing $\alpha$.

 $\mathsf{Setup}$. $\mathcal{B}$ initializes the corrupted party set 
  $\boldsymbol{F}=\emptyset$ and the signature query set $\boldsymbol{Q}=\emptyset$,
  and sends $pp=(p,\mathbb{G},g,h)$ to $\mathcal{A}$.

   $\mathsf{KeyGen}$. $\mathcal{B}$ picks a random index $i^*\in [1,q_k]$. 
  $\mathcal{B}$ runs $\mathbf{KeyGen}(pp) \to (pk_i,sk_i)$, 
  for $i \in [1,q_k],i\ne i^*$, sets $pk_{i^*}=g^\alpha$. 
  By doing this, $\mathcal{B}$ implicitly defines $sk_{i^*}=\alpha$.
  $\mathcal{B}$ sets $\boldsymbol{P}=(pk_1,pk_2,\dots,pk_{q_k})$, 
  and sends $\boldsymbol{P}$ and $\boldsymbol{W}$ to $\mathcal{A}$.

 $\mathsf{Query}$. The procedure for $\mathcal{B}$ to answer the adversary's oracle queries is consistent with the responses to the $\mathcal{HO}$, $\mathcal{CO}$, $\mathcal{PSO}$ and $\mathcal{SO}$ oracles in the proof of Theorem~\ref{thm:unforgeability}.

 $\mathsf{Challenge}$. $\mathcal{A}$ successfully outputs two signatures $\left(\boldsymbol{PK_i^*},\sigma_i^*,m_i^*\right)$, 
where $\sigma_i^*=(z_i^*,\boldsymbol{C}_i^*,\boldsymbol{TAG}_i^*)$, and $i\in$ $\{1,2\}$, with the following conditions:
	\begin{itemize}[]
		\item $\boldsymbol{PK}_i^*\subseteq \boldsymbol{P}$ and $\left| \boldsymbol{PK}_i^* \right|=n$, for each $i\in \{1,2\}$,
		\item $|(\boldsymbol{PK}_1^*\cup\boldsymbol{PK}_2^*) \cap \boldsymbol{F}|\le 2t-1$,
		\item $(\boldsymbol{PK}_i^*,*,m_i^*) \notin \boldsymbol{Q}$, for each $i\in \{1,2\}$,
		\item $\mathbf{Verify}(pp,\boldsymbol{PK}_i^*,\sigma_i^*,t,m_i^*)=1$, for each $i\in \{1,2\}$, 
		\item $\mathbf{Link}(\sigma_1^*,\sigma_2^*)=0$.
	\end{itemize}
If $pk_{i^*}\notin \boldsymbol{PK}_1^*\cup\boldsymbol{PK}_2^*$, $\mathcal{B}$ declares failure and aborts.
Otherwise, $\mathcal{B}$ computes $\alpha$ as follows.

 Since $\sigma_1^*$ and $\sigma_2^*$ are unlinked, for any $i, j \in [0, t-1]$, the tag $tag_i \in \boldsymbol{TAG_1^*}$ 
is not equal to the tag $tag_j \in \boldsymbol{TAG_2^*}$. This implies that the two signatures collectively generate $2t$ distinct 
tags (i.e., utilize $2t$ different private keys). In the set of public keys $\boldsymbol{PK}_1^*\cup\boldsymbol{PK}_2^*$, 
the adversary $\mathcal{A}$ obtains at most $2t - 1$ private keys.
Given that both signatures are valid $(t, n)$-threshold signatures, it follows that at least one signature 
must be a successful forgery by the adversary. According to Theorem~\ref{thm:unforgeability}, 
the simulator can compute $\alpha$ to break the DL assumption.

$\mathsf{Probability\ Analysis}$. The analysis parallels that of Theorem~\ref{thm:unforgeability}. Assuming $\mathcal{A}$'s forgery success probability is $\epsilon(\lambda)$, and given the condition $p>q_{ps}(q_h+q_{ps}-1)+q_{s}(q_h+q_{s}-1)$ (with $q_h$, $q_s$, $q_{ps}$ representing the numbers of hash, signature, and pre-signature queries), the probability $\epsilon'(\lambda)$ for $\mathcal{B}$ to solve the DL problem is at least $\frac{2n-2t+1}{32nq_k}\epsilon(\lambda)$.
\end{proof}

\section{Performance Evaluation}
\label{sec:performance}
\subsection{Experiment Setup}
The simulation experiment is performed on a laptop with a
1.90 GHz Intel Core i5-1340P CPU, 16GB memory, and Windows
11 operating system. For programming, we use
Java language and the jdk version 1.8 \cite{OracleJavaSE8}.

\subsection{Experiment Evaluation}
Figure \ref{fig:runtime_comparison} illustrates a performance comparison of the proposed LTRAS scheme relative to the baseline scheme of \cite{wang2024anonymity}, which is a linkable dualring adaptor signature based on a schnorr instantiation. The evaluation assesses the computational efficiency of 9 distinct algorithms as the ring size scales from 10 to 100 members, with the signature threshold set to half of the ring size. 

The proposed LTRAS scheme exhibits significantly superior efficiency in the critical and frequently executed algorithms of signature generation and verification. As shown in Figure \ref{fig:PreSign}, \ref{fig:PreVerify}, and \ref{fig:Verify}, the computation cost for these core algorithms in the LTRAS scheme is substantially lower and scales more gracefully with increasing ring size compared to the baseline. This superior scalability is crucial for practical deployment.
Conversely, Figures~\ref{fig:Setup}, \ref{fig:GenR}, and \ref{fig:Ext} reveal that, the LTRAS scheme demands slightly more computation in the $\mathbf{Setup}$, $\mathbf{GenR}$, and $\mathbf{Ext}$ algorithms. However, these algorithms are typically performed only once during initialization or infrequently during the system's lifecycle. 

Figure \ref{fig:KeyGen} reveals that the computation cost for the $\mathbf{KeyGen}$ algorithm of both schemes are closely aligned. In contrast, Figure \ref{fig:Adapt} indicates that the $\mathbf{Adapt}$ algorithm in the LTRAS scheme maintains a constant, negligible computation cost, independent of the ring size, while the baseline scheme incurs a cost that scales with the ring size. As illustrated in Figure \ref{fig:Link}, both schemes demonstrate comparable efficiency in the $\mathbf{Link}$ algorithm. The results for the $\mathbf{Link}$ algorithm in Figure \ref{fig:Link} show that our scheme maintains consistently lower latency compared to the baseline scheme across all tested ring sizes. 

We analyze the communication cost of the proposed scheme relative to \cite{wang2024anonymity}. Let $|\mathbb{G}_p|$ represent the size of a group element in the cyclic group $\mathbb{G}_p$, and $|\mathbb{Z}_p|$ indicate the size of an element in the finite field $\mathbb{Z}_p$. As compiled in Table \ref{tab:communication_cost}, the comparative findings show that although our scheme entails moderately higher expenses in $\mathbf{Setup}$ and $\mathbf{GenR}$, it demonstrates marked efficiency in critical algorithms.
Specifically, the cost of the baseline scheme for $\mathbf{PreSign}$ and $\mathbf{Adapt}$ algorithms scales as $t(n+1)|\mathbb{Z}_p|+t|\mathbb{G}_p|$, growing both with the threshold $t$ and the ring size $n$. In contrast, our scheme reduces this term to $(n+1)|\mathbb{Z}_p|+t|\mathbb{G}_p|$, eliminating the multiplicative factor $t$ from the $\mathbb{Z}_p$ component. This optimization significantly lowers the communication overhead for these frequently executed algorithms, enhancing overall system efficiency.

\begin{figure*}[h!]
    \centering
	\captionsetup[subfigure]{font=small}
    \begin{subfigure}[b]{0.32\textwidth}
        \centering
        \includegraphics[width=\linewidth]{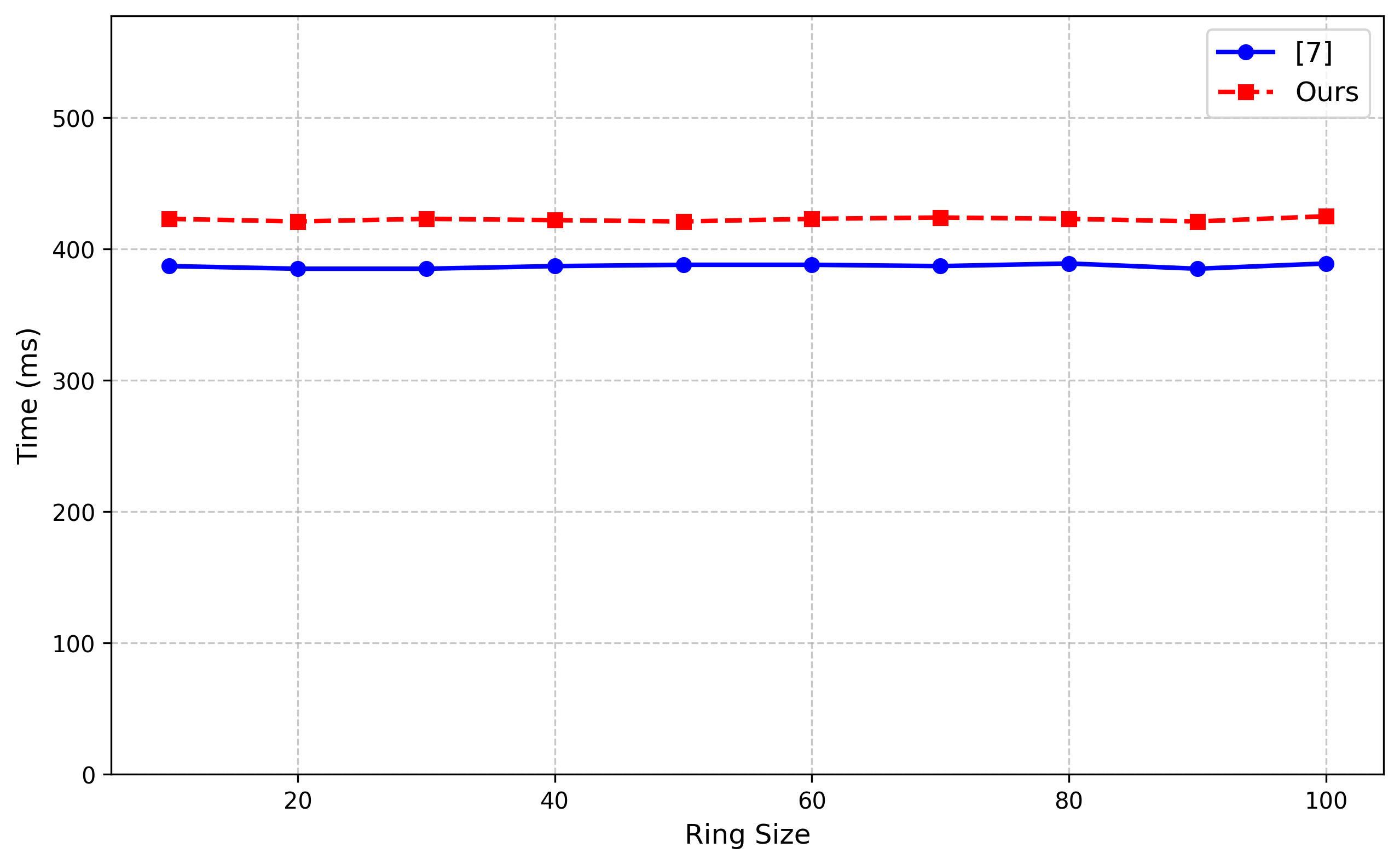}
        \caption{$\mathbf{Setup}$}
        \label{fig:Setup}
    \end{subfigure}
    \hfill
    \begin{subfigure}[b]{0.32\textwidth}
        \centering
        \includegraphics[width=\linewidth]{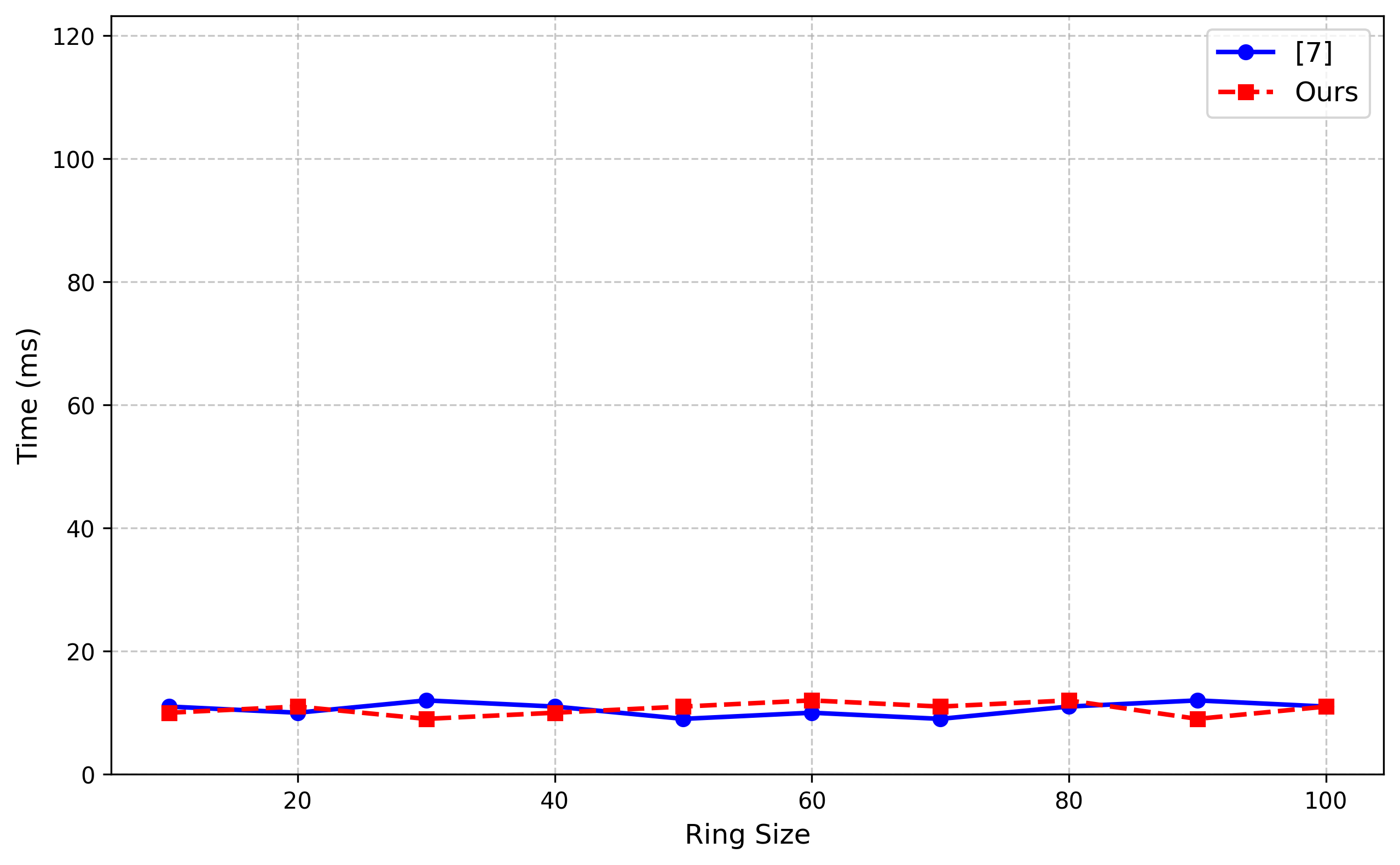}
        \caption{$\mathbf{KeyGen}$}
        \label{fig:KeyGen}
    \end{subfigure}
    \hfill
    \begin{subfigure}[b]{0.32\textwidth}
        \centering
        \includegraphics[width=\linewidth]{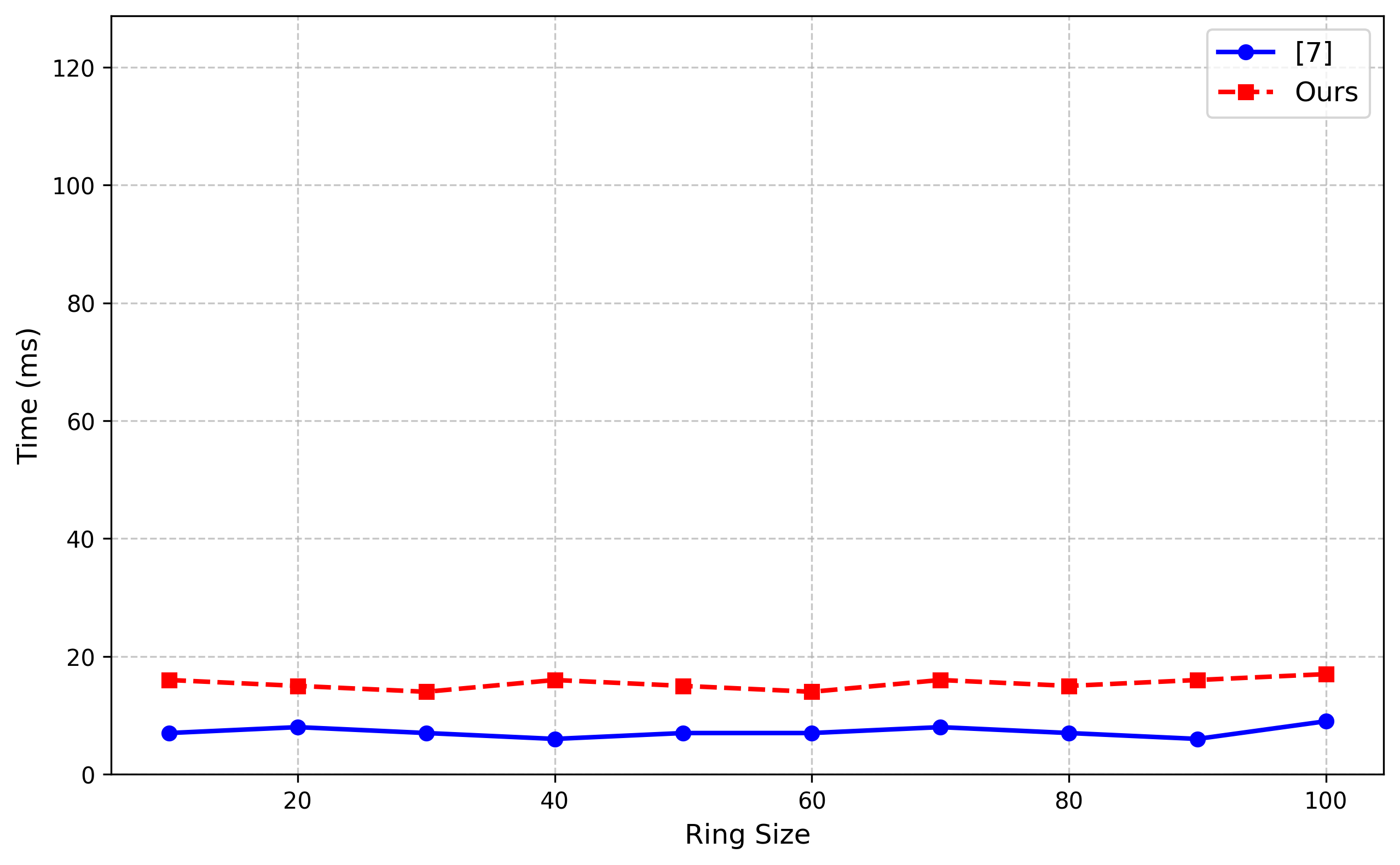}
        \caption{$\mathbf{GenR}$}
        \label{fig:GenR}
    \end{subfigure}

    \begin{subfigure}[b]{0.32\textwidth}
        \centering
        \includegraphics[width=\linewidth]{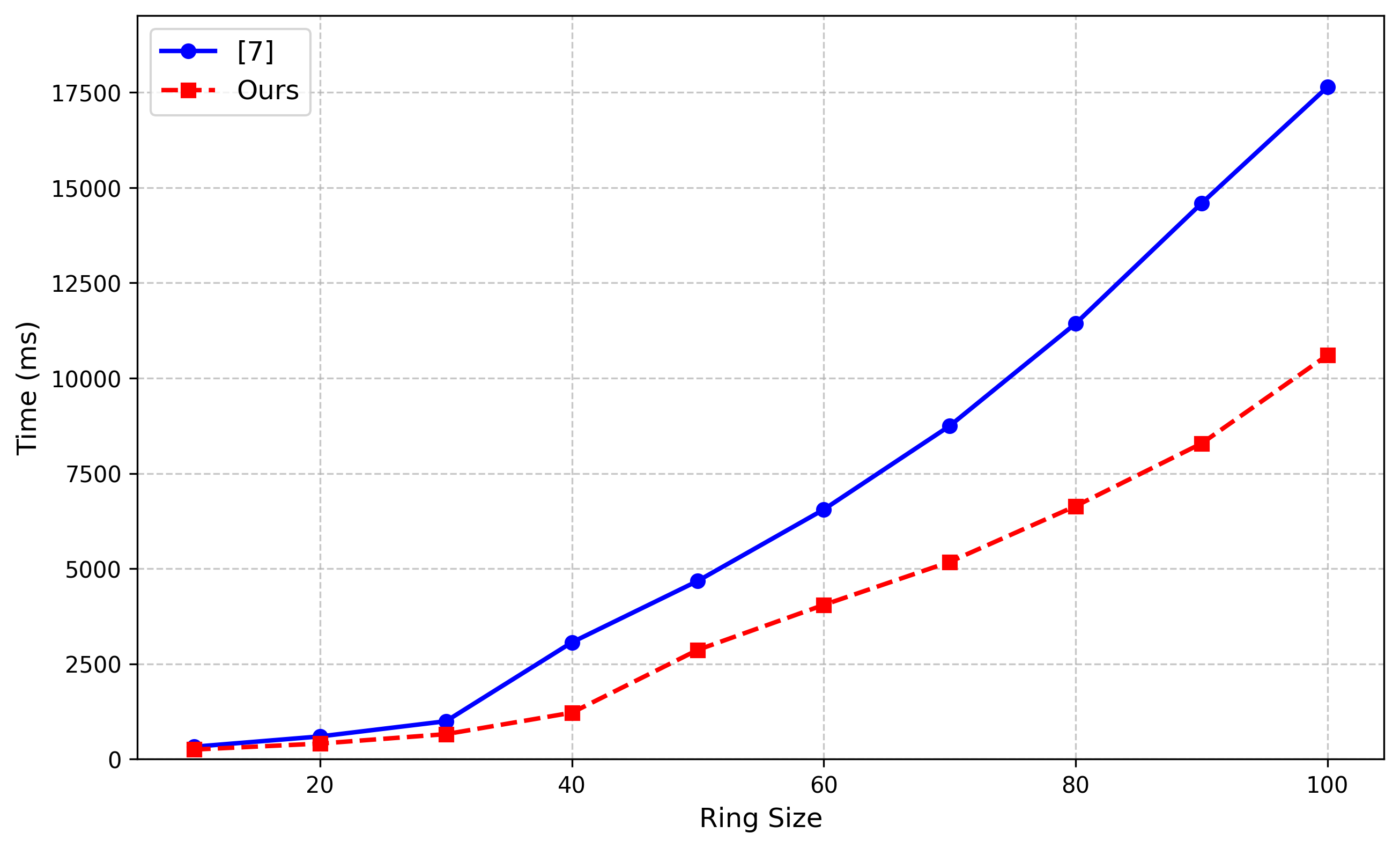}
        \caption{$\mathbf{PreSign}$}
        \label{fig:PreSign}
    \end{subfigure}
    \hfill
	\begin{subfigure}[b]{0.32\textwidth}
        \centering
        \includegraphics[width=\linewidth]{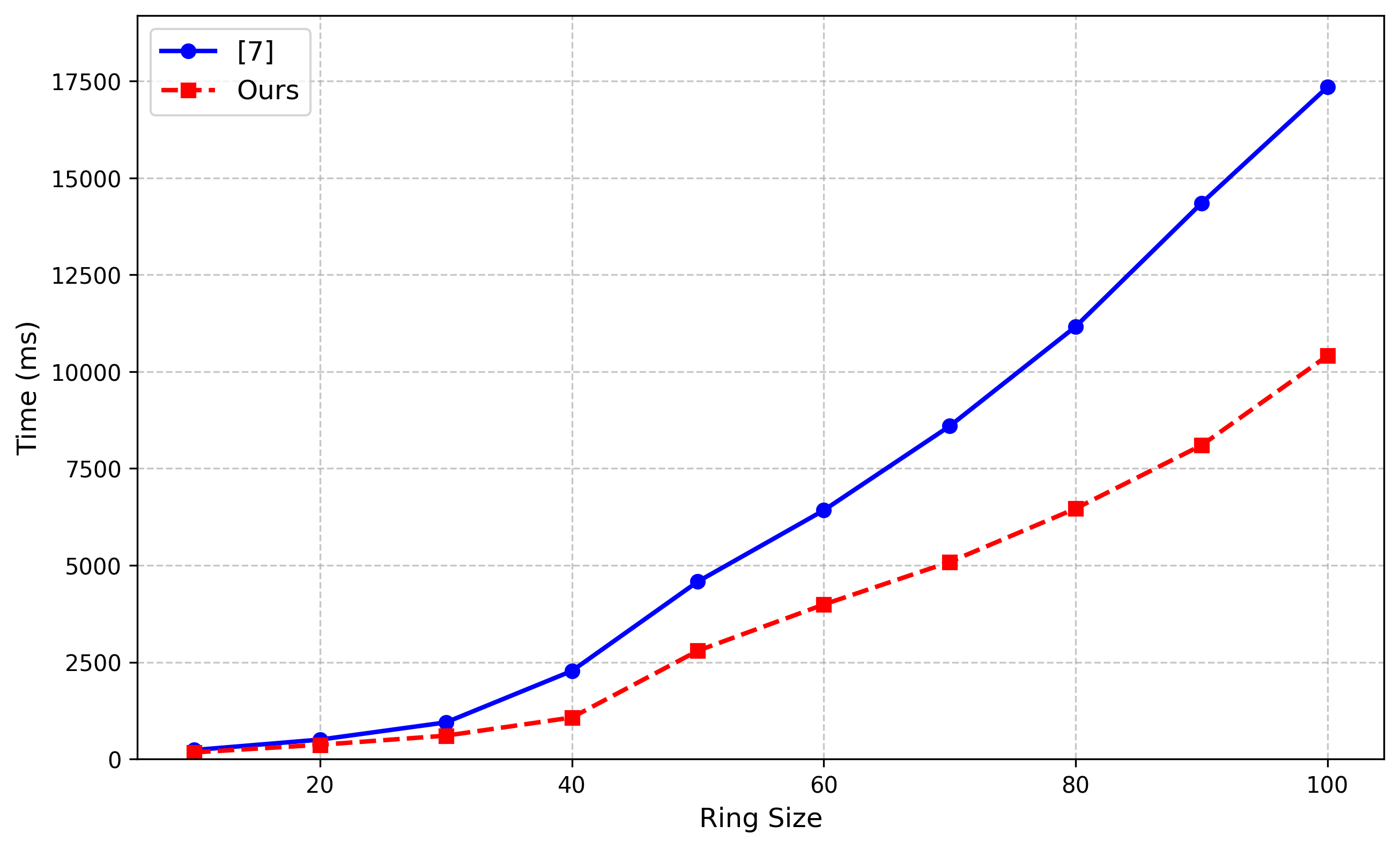}
        \caption{$\mathbf{PreVerify}$}
        \label{fig:PreVerify}
    \end{subfigure}
    \hfill
	\begin{subfigure}[b]{0.32\textwidth}
        \centering
        \includegraphics[width=\linewidth]{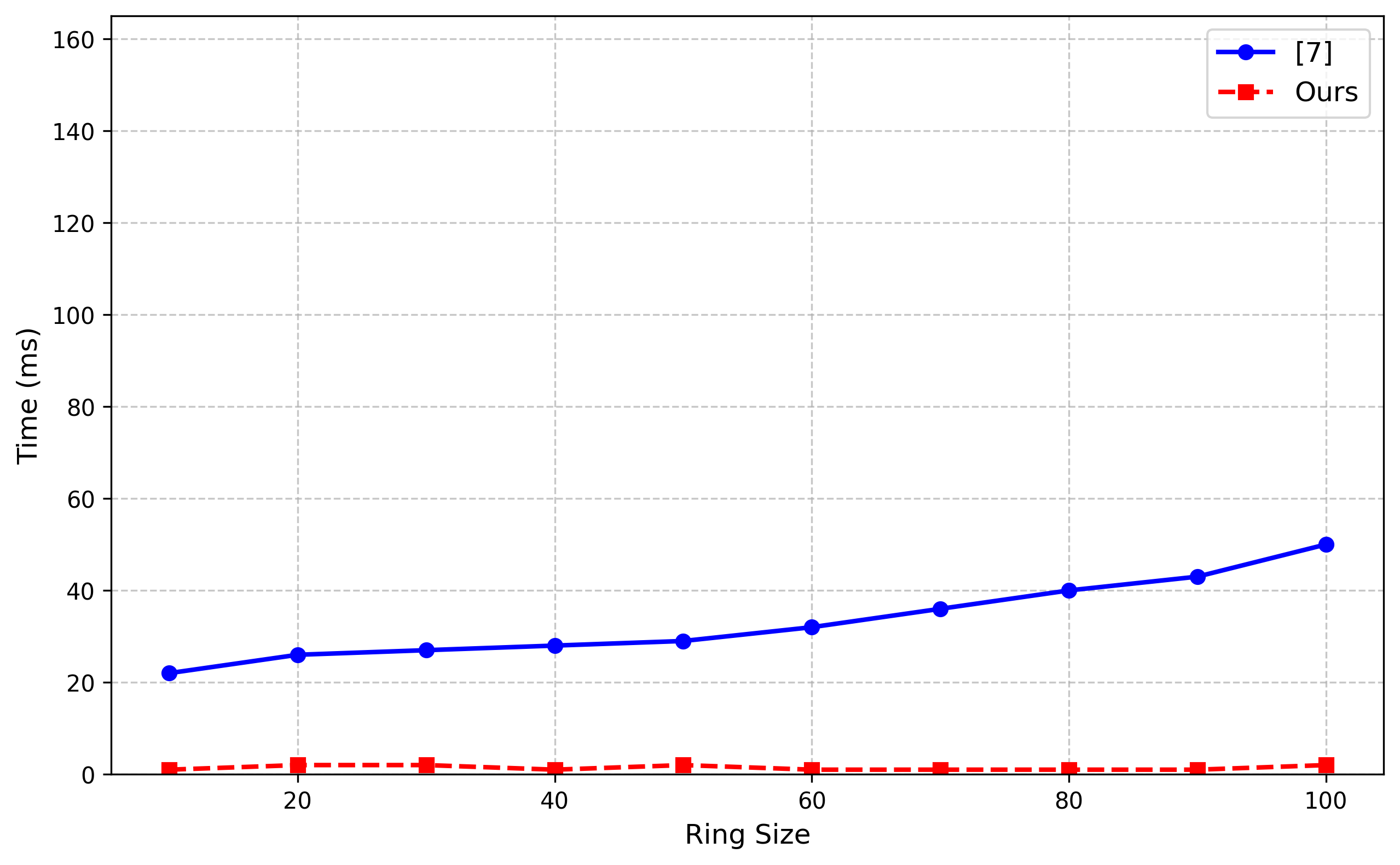}
        \caption{$\mathbf{Adapt}$}
        \label{fig:Adapt}
    \end{subfigure}
    \hfill
    
    \begin{subfigure}[b]{0.32\textwidth}
        \centering
        \includegraphics[width=\linewidth]{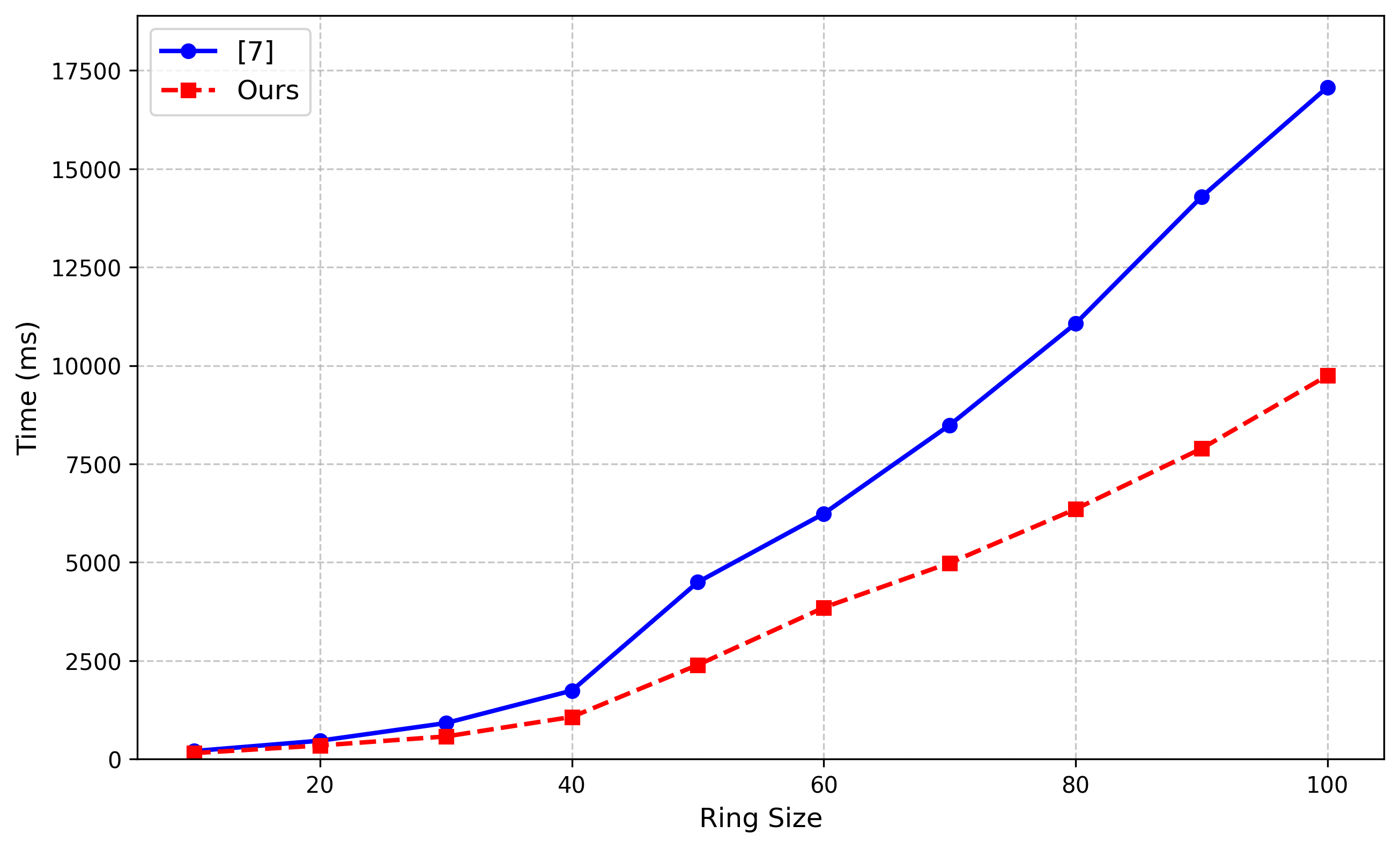}
        \caption{$\mathbf{Verify}$}
        \label{fig:Verify}
    \end{subfigure}
    \hfill
    \begin{subfigure}[b]{0.32\textwidth}
        \centering
        \includegraphics[width=\linewidth]{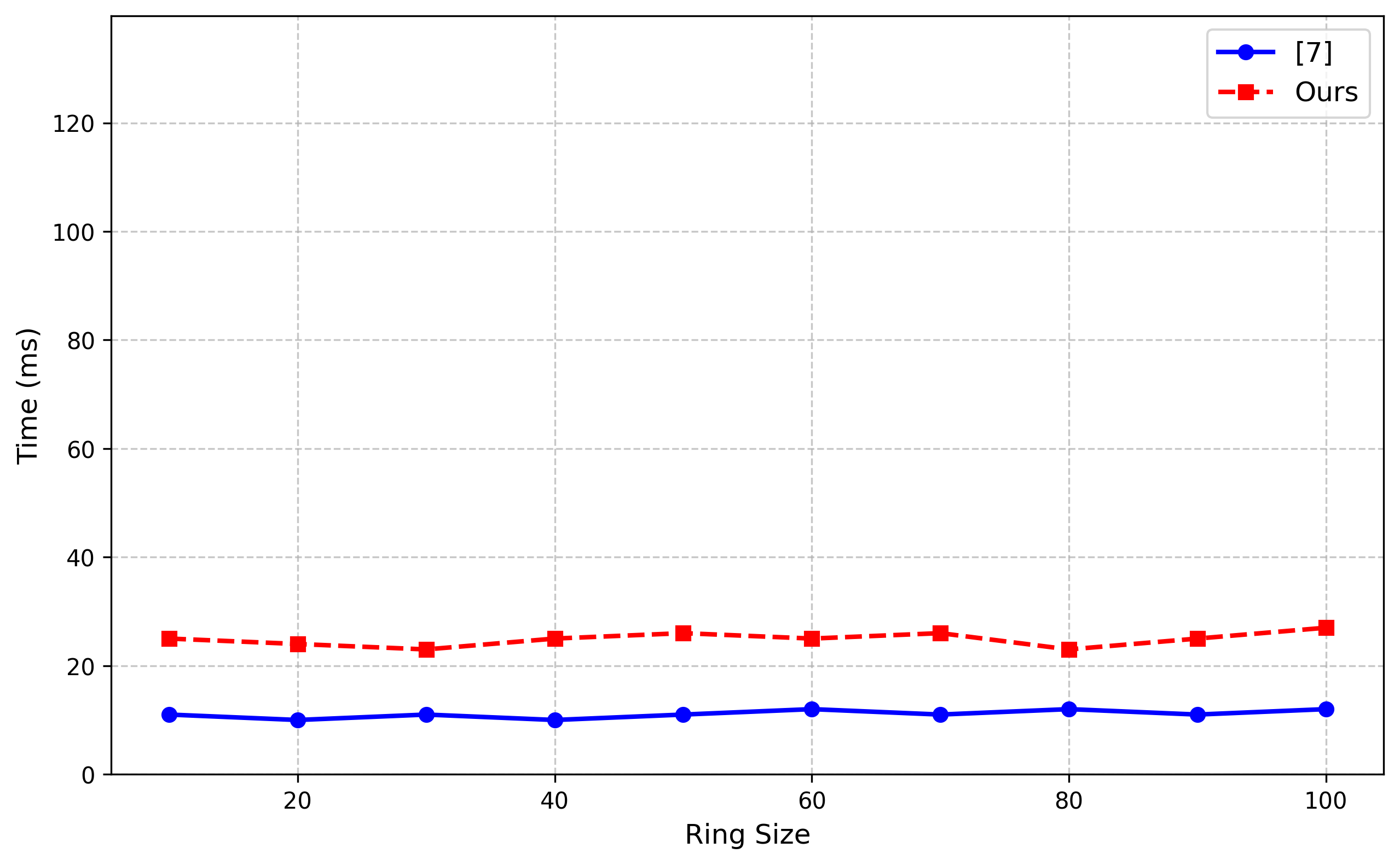}
        \caption{$\mathbf{Ext}$}
        \label{fig:Ext}
    \end{subfigure}
    \hfill
	\begin{subfigure}[b]{0.32\textwidth}
        \centering
        \includegraphics[width=\linewidth]{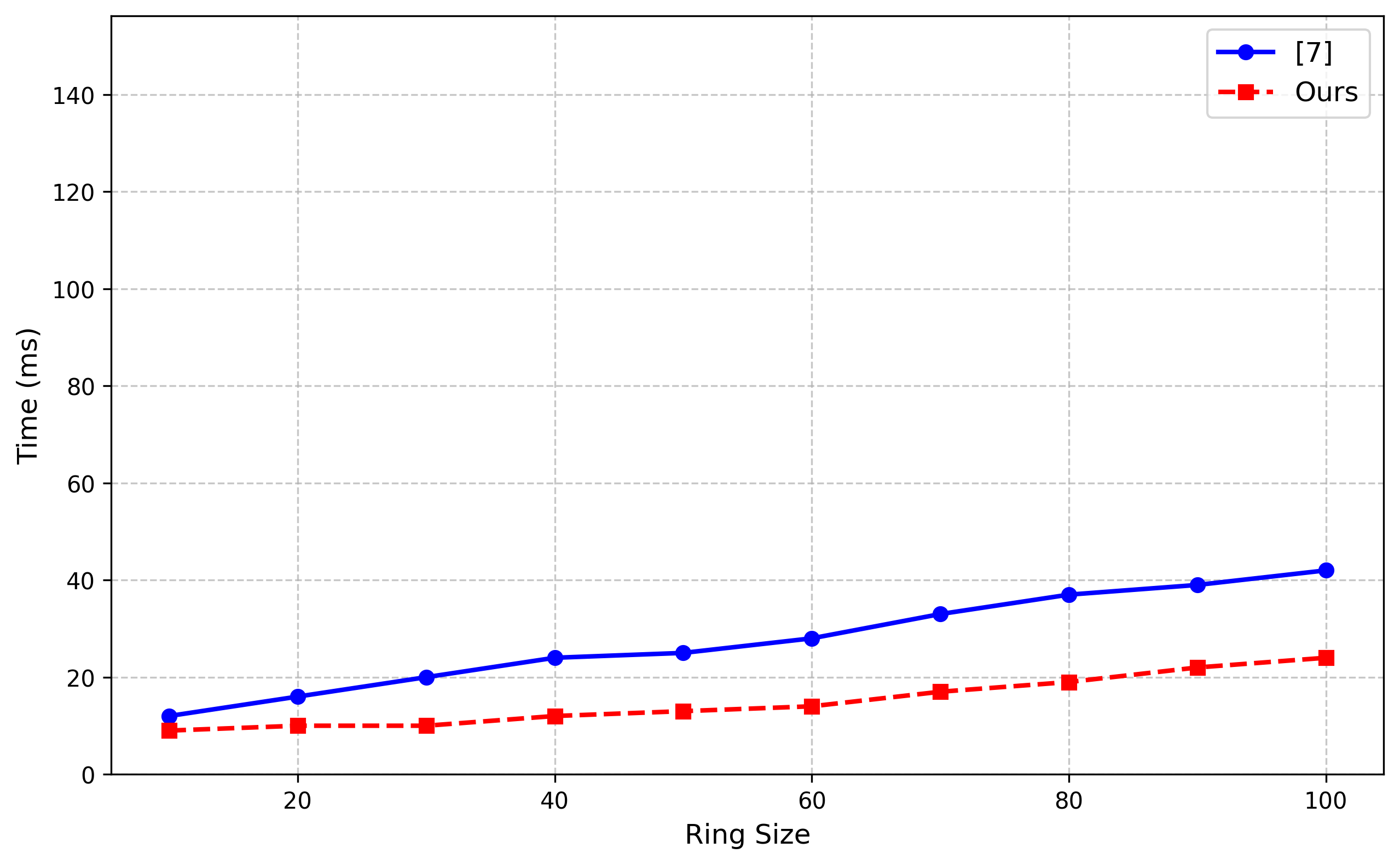}
        \caption{$\mathbf{Link}$}
        \label{fig:Link}
    \end{subfigure}
    \hfill

    \caption{Runtime comparison of different algorithms with varying ring sizes.}
    \label{fig:runtime_comparison}
\end{figure*}

\begin{table*}[h!]
\centering
\caption{Communication cost comparison of different algorithms($n$ is the ring size, $t$ is the threshold size)}
\label{tab:communication_cost}
\begin{tabular}{lccccc}
\toprule
\multirow{2}{*}{Schemes} & \multicolumn{5}{c}{Algorithms} \\
\cmidrule{2-6}
& $\mathbf{Setup}$ & $\mathbf{KeyGen}$ & $\mathbf{GenR}$ & $\mathbf{PreSign}$ & $\mathbf{Adapt}$ \\
\midrule
\cite{wang2024anonymity}   & $|\mathbb{G}_p|$ & $n|\mathbb{G}_p|$ & $|\mathbb{G}_p|$ & $t(n+1)|\mathbb{Z}_p|+t|\mathbb{G}_p|$ & $t(n+1)|\mathbb{Z}_p|+(t+1)|\mathbb{G}_p|$ \\
Ours   & $2|\mathbb{G}_p|$  & $n|\mathbb{G}_p|$  & $2|\mathbb{G}_p|$  & $(n+1)|\mathbb{Z}_p|+t|\mathbb{G}_p|$  & $(n+1)|\mathbb{Z}_p|+t|\mathbb{G}_p|$ \\
\bottomrule
\end{tabular}
\end{table*}

\section{Application}
\label{sec:application}
\begin{figure*}[h!]
  \centering
  \includegraphics[width=0.8\textwidth]{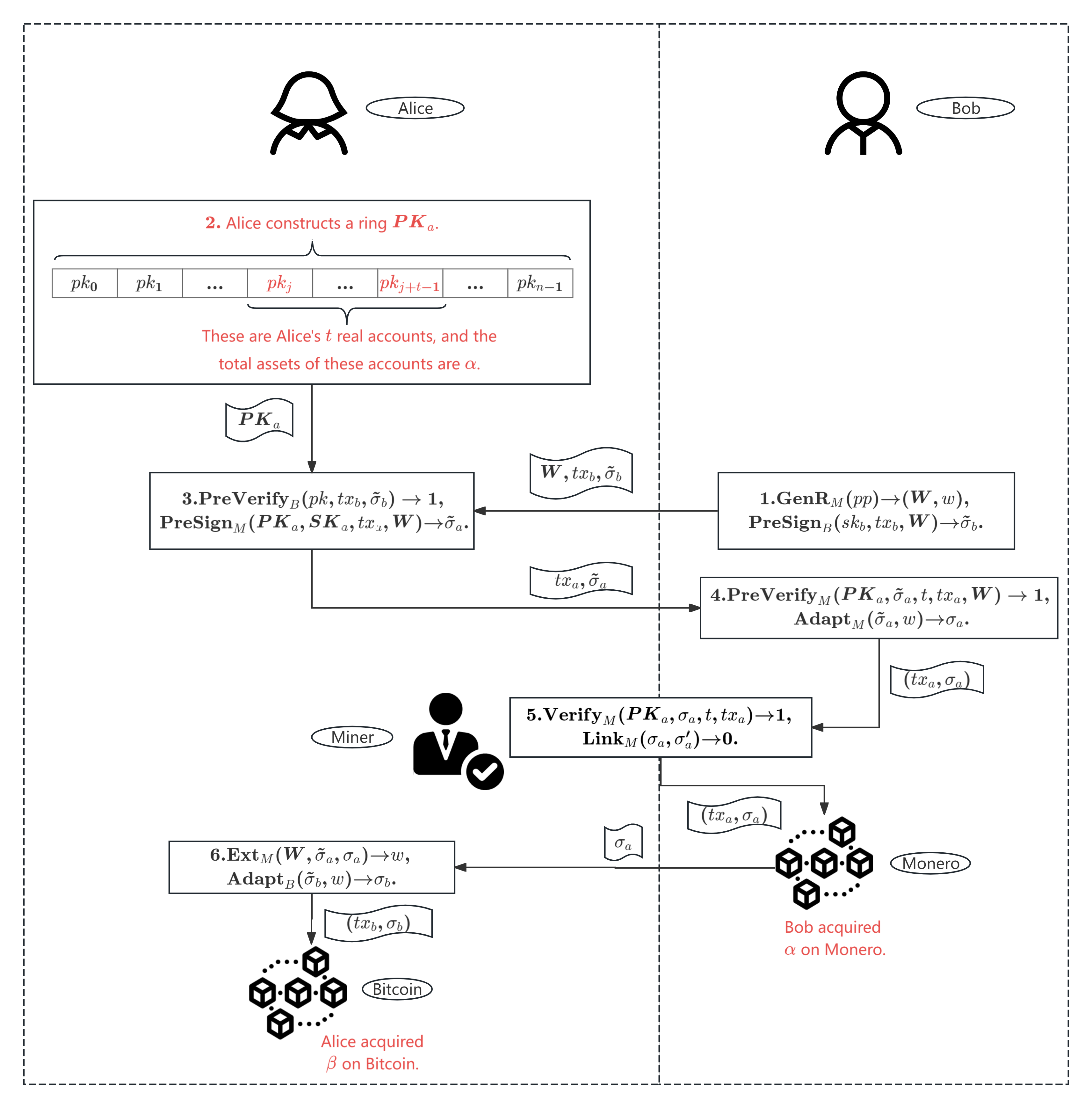}
  \caption{Cross chain atomic swaps between Bitcoin and Monero.}
  \label{fig:atomic_swap}
\end{figure*}

This section explores the practical applications of 
the proposed scheme, focusing on its usefulnessin multi-account joint payment systems. 

Consider the following cross-chain asset swap scenario: 
Alice holds multiple accounts on Monero with a total asset value of $\alpha$, 
and Bob holds assets valued at $\beta$ on Bitcoin. They 
wish to swap assets with each other, meaning Alice 
wants to acquire assets worth $\beta$ on Bitcoin, 
while Bob wants to acquire assets worth $\alpha$ on Monero. 
A crucial property that this swap must satisfy is 
atomicity (also referred to as fairness), i.e., 
the swap terminates in one of two mutually exclusive states: a state where both parties obtain the other's assets, or a state where neither does.

In the aforementioned scenario, Bob can utilize existing 
ECDSA adaptor signatures \cite{aumayr2021generalized} to 
complete his part of the process. However, existing adaptor 
signature schemes are insufficient for Alice to complete 
the process securely. Firstly, for privacy reasons, Alice wants to keep her 
account addresses confidential. This necessitates an adaptor signature 
scheme capable of protecting the privacy of account 
addresses. Secondly, Alice's total assets are distributed 
across multiple accounts, requiring the use of all these 
accounts as inputs during payment. Therefore, an adaptor 
signature scheme supporting multi-account joint payment 
is needed. Furthermore, while protecting the privacy of 
the signer's account addresses, the scheme must also 
achieve linkability of signatures to prevent malicious 
users from using the same set of input addresses to 
generate two different transactions, thereby enabling a 
double-spending attack.

Our proposed LTRAS scheme comprises several integral components that meet the specified requirements. First, the system invokes the $\mathbf{Setup}$ algorithm 
to generate public parameters. Users in the system invoke the $\mathbf{KeyGen}$ algorithm to derive their public-private key pairs. The public key is designated as the user's account address, and a single user may control multiple such key pairs (and thus multiple addresses). This scheme employs a 
$(t, n)$-threshold, meaning signatures require the 
participation of private keys from at least $t$ signers; 
fewer than $t$ keys cannot generate a valid signature. 
Subsequently, Alice and Bob participate in the interactive protocol illustrated in Figure~\ref{fig:atomic_swap}. In this context, the subscript '$B$' refers to operations in the Bitcoin adaptor signature scheme, while '$M$' refers to those in our LTRAS scheme applied to Monero. The protocol proceeds as follows:
\begin{enumerate}
	\item Bob executes the hard relation generation algorithm $\mathbf{GenR}_M$ to produce a statement $\boldsymbol{W}$ and its corresponding witness $w$. He then prepares the transaction data structure $tx_b$ for an asset transfer to Alice on the Bitcoin network. He then obtains a corresponding pre-signature $\tilde{\sigma}_b$ via $\mathbf{PreSign}_B$. Bob sends $\boldsymbol{W}$, $tx_b$, and $\tilde{\sigma}_b$ to Alice, keeping $w$ locally.
	\item Alice first conceals her $t$ accounts 
	$\left(pk_j, ...,\right.$ $\left.\\ pk_{j+t-1}\right)$ within a ring $\boldsymbol{PK}_a$ 
	that includes accounts from other users. The private 
	keys corresponding to these $t$ accounts are denoted $\boldsymbol{SK_a}$.
	\item Alice invokes the pre-signature verification 
	algorithm $\mathbf{PreVerify}_B$ on Bitcoin to validate 
	$\tilde{\sigma}_b$. If verification fails, she 
	terminates the protocol. Subsequently, Alice constructs 
	the transaction data structure $tx_a$ for transferring 
	assets to Bob on Monero. She then invokes the 
	pre-signature generation algorithm $\mathbf{PreSign}_M$ 
	on Monero to generate a pre-signature $\tilde{\sigma}_a$ 
	for this transaction. Alice sends $\tilde{\sigma}_a$ 
	and $tx_a$ to Bob. The transaction cannot be settled instantly, since the pre-signature lacks validity as a blockchain signature.
	\item Bob invokes the pre-signature verification algorithm 
	$\mathbf{PreVerify}_M$ on Monero to validate 
	$\tilde{\sigma}_a$. If verification fails, he 
	terminates the protocol. As Bob possesses the witness 
	$w$ to the hard relation, he can, by the adaptability 
	property of adaptor signatures, invoke the adapt 
	algorithm $\mathbf{Adapt}_M$ to convert the 
	pre-signature $\tilde{\sigma}_a$ received from 
	Alice into a full, blockchain-valid signature 
	$\sigma_a$. He then broadcasts $\sigma_a$ and the 
	transaction $tx_a$ to Monero miners.
	\item Miners invoke the signature verification algorithm 
	$\mathbf{Verify}_M$ to validate $\sigma_a$ and the 
	linking algorithm $\mathbf{Link}_M$ to determine if 
	the transaction shares a link with any previous 
	transactions, thereby detecting potential double-spending. 
	If the signature is valid and no double-spending is 
	detected, the transaction is considered valid. 
	The miners can then include the transaction and its 
	signature in a block. At this point, Bob successfully 
	acquires the assets $\alpha$ on Monero.
	\item After observing this transaction and its signature 
	$\sigma_a$ on the blockchain, Alice, leveraging the 
	witness extractability property of adaptor signatures, 
	invokes the witness extraction algorithm $\mathbf{Ext}_M$. 
	Using the full signature $\sigma_a$ from the chain and 
	her originally generated pre-signature $\tilde{\sigma}_a$, 
	she extracts the witness $w$. Subsequently, based on 
	the adaptability property, Alice can use the witness $w$ 
	and invoke the adapt algorithm $\mathbf{Adapt}_B$ to 
	convert the pre-signature $\tilde{\sigma}_b$ received 
	earlier from Bob into a full signature $\sigma_b$. She 
	then broadcasts $\sigma_b$ and $tx_b$ for verification 
	by Bitcoin miners. Once Bob's transfer transaction 
	$tx_b$ is confirmed, Alice finally acquires the assets $\beta$ on Bitcoin.
\end{enumerate}

This protocol ensures atomicity in the asset swap between
Alice and Bob. If either party fails to complete their
part of the transaction, neither will receive the other's
assets. Additionally, Alice's use of multiple accounts
is seamlessly integrated into the process, with her
privacy preserved through the ring structure, and
double-spending prevented via the linkability feature
of the LTRAS scheme.
\section{Conclusion}
\label{sec:conclusion}
This paper has introduced a novel cryptographic scheme termed LTRAS. The proposed scheme effectively balances the requirements of privacy, atomicity, and efficiency within blockchain protocols.
It uniquely enables efficient, atomic transactions that preserve payer anonymity even when funds are consolidated from multiple input accounts in UTXO-based chains that has been largely unaddressed by prior work.

We proposed a formal definition and security model, and proved that LTRAS satisfies the essential properties of correctness, pre-signature adaptability, witness extractability, existential unforgeability, anonymity, and linkability. The practical viability of LTRAS is demonstrated through a comprehensive experimental evaluation. The results confirm that LTRAS achieves significant performance improvements over existing schemes, particularly in signature generation and verification, which are critical for real-world applications.

Future work would be to remove the adjacency constraint for real accounts within the ring, thereby enhancing anonymity without compromising efficiency. Finally, an instantiation based on post-quantum cryptographic assumptions would be a crucial step toward ensuring long-term security.

\section{Acknowledgement}
This work was supported in part by the National Natural Science Foundation of China under Grant 62372103, the Natural Science Foundation of Jiangsu Province under Grant BK20231149, and the Jiangsu Provincial Scientific Research Center of Applied Mathematics under Grant BK20233002.


\bibliographystyle{cas-model2-names}
\bibliography{cas-refs}



\subsection*{  } 
\setlength\intextsep{0pt} 
\begin{wrapfigure}{l}{25mm}
    \centering
    \includegraphics[width=1in,height=1.25in,clip,keepaspectratio]{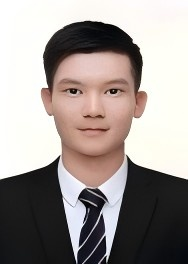}
\end{wrapfigure}
\noindent \textbf{Yi Liang} Yi Liang received the B.E. degree in the School of Computer Science and Technology from Anhui University, Hefei, China, in 2023. 
He is currently working toward the M.E. degree with Southeast University, Nanjing, China. 
His current research interests include applied cryptography and blockchain privacy protection.\par

\hspace*{\fill} 

\setlength\intextsep{0pt} 
\begin{wrapfigure}{l}{25mm}
    \centering
    \includegraphics[width=1in,height=1.25in,clip,keepaspectratio]{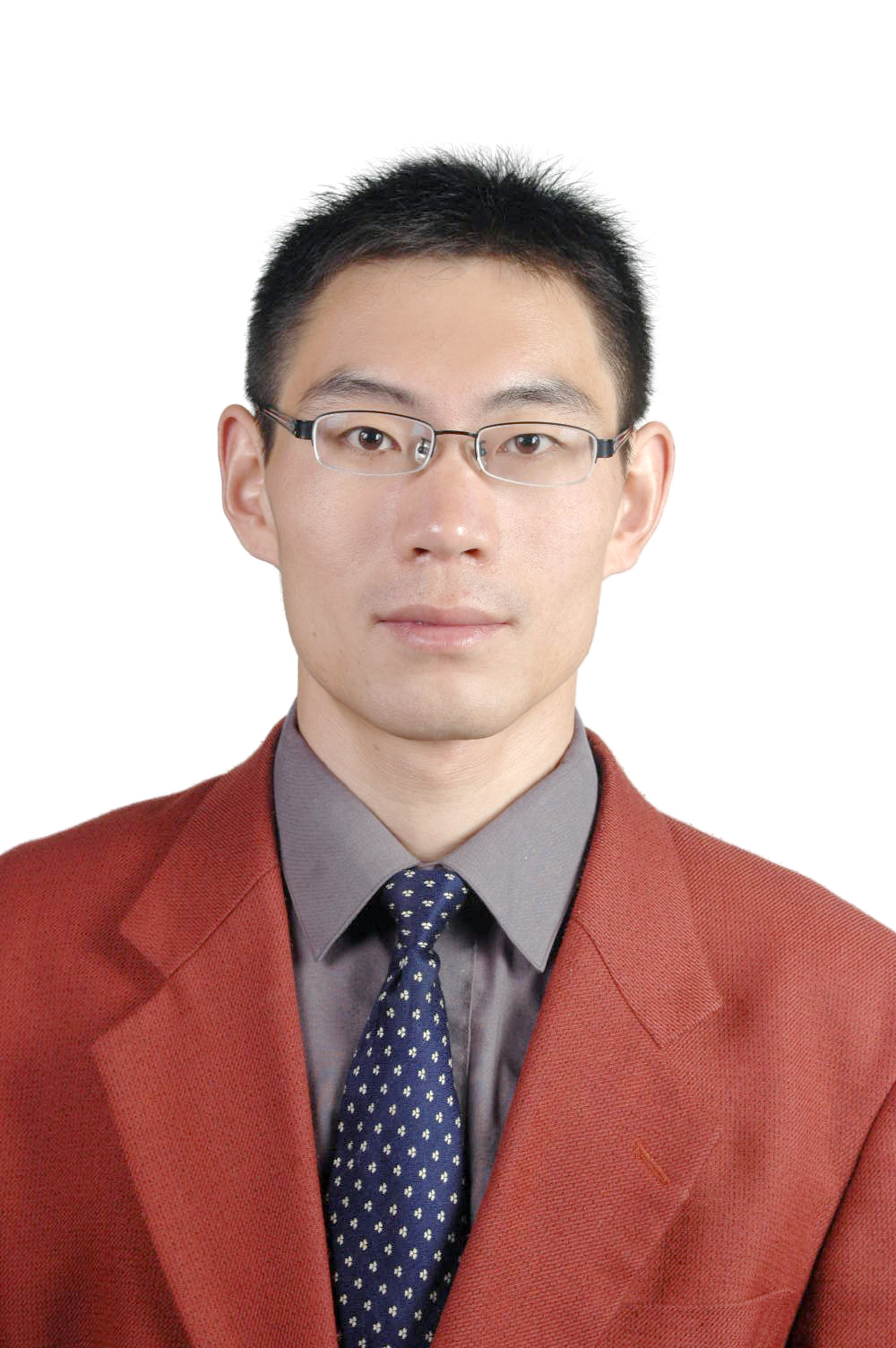}
\end{wrapfigure}
\noindent \textbf{Jinguang Han} Jinguang Han received the Ph.D. degree from the University of Wollongong, Australia, in 2013. He is a Professor with the School of Cyber Science and Engineering, Southeast University, China. His research focuses on access control, cryptography, cloud computing, and privacy-preserving systems.\par

\hspace*{\fill} 

\end{document}